\documentclass[twoside]{article}
\usepackage[a4paper]{geometry}
\usepackage[latin1]{inputenc} 
\usepackage[T1]{fontenc} 
\usepackage{RR}
\usepackage{hyperref}

\usepackage{subfigure}
\usepackage{graphicx}
\usepackage{enumitem}
\usepackage{url}
\usepackage{comment}
\usepackage{amssymb}
\usepackage{mathenv}
\usepackage{amsmath}
\usepackage{amsthm}

\usepackage{algorithm, algorithmic}  
\usepackage{comment}
\usepackage{subfigure}

\usepackage{array}
\makeatletter
\@addtoreset{section}{part}
\makeatother  

\DeclareMathOperator*{\argmax}{arg\,max}
\RRNo{XXX}
\RRdate{September 2014}

\graphicspath{{./logos/}{figures/}{figures/fig++/}}

\RRauthor{
Ichrak Amdouni, 
Cedric Adjih, 
Pascale Minet
}
\authorhead{amdouni \& 	adjih \& minet}
\RRtitle{Routage et Ordonnancement STDMA pour la minimisation des D\'elais dans les R\'eseaux de Capteurs sans fil en Grille}
\RRetitle{Joint Routing and STDMA-based Scheduling to Minimize Delays in Grid Wireless Sensor Networks}
\titlehead{ORCHID}
\RRresume{
Dans ce rapport , nous \'etudions l'optimisation des d\'elais et l'efficacit\'e \'energ\'etique dans
r\'eseau des r\'eseaux de capteurs sans fil (WSNs). Nous nous concentrons sur l'ordonnancement STDMA (Spatial Reuse TDMA)). STDMA se base sur la r\'ep\'etition d'un cycle o\`u chaque noeud a
un slot pour transmettre ses donn\'ees (ces slots sont d\'efinis par les $couleurs$). Nous supposons un algorithme d'ordonnancement STDMA qui tire avantage de la r\'egularit\'e de la topologie de r\'eseau en grille pour fournir \'egalement un coloriage spatialement p\'eriodique (obtenu en r\'ep\'etant un certain motif de couleurs). 

Dans ce cadre, les principaux d\'efis sont les suivants: 1) r\'eduire au minimum les d\'elais de routage en ordonnant les slots sur le cycle de fa\c con efficace. 
2) tout en \'etant \'econome en \'energie. Notre travail suit deux directions: d'abord, la performance de solutions de base c'est \`a dire, quand rien de pr\'ecis est fait et les couleurs sont ordonn\'ees arbitrairement sur le cycle STDMA. Ensuite, nous proposons une solution appel\'ee ORCHID ("Optimized Routing and sCHeduling in grID wireless sensor networks") qui d\'elib\'er\'ement construit un 
ordonnancement STDMA efficace. 
ORCHID proc\`ede en deux \'etapes. Dans la premi\`ere \'etape, ORCHID part d'une grille color\'ee et construit un routage hi\'erarchique bas\'e sur ces couleurs. Dans la deuxi\`eme \'etape, ORCHID ordonne les couleurs sur le cycle STDMA en consid\'erant conjointement le routage et l'ordonnancement dans le but de permettre \`a tout noeud d'atteindre le puits en un seul cycle.
Nous \'etudions les performances de ces solutions par mod\'elisation et simulation.
Les r\'esultats montrent les performances excellentes de ORCHID en termes d'\'economie d'\'energie et des d\'elais compar\'e au routage bas\'e sur les plus courts chemins et utilisant le d\'elai comme m\'etrique. Nous pr\'esentons aussi l'adaptation de ORCHID aux graphes quelconques sous le mod\`ele SINR.

}

\RRabstract{
In this report, we study the issue of delay optimization and energy 
efficiency in grid wireless sensor networks (WSNs).
We focus on STDMA (Spatial Reuse TDMA)) scheduling, where a predefined cycle is repeated,
and where each node has fixed transmission opportunities 
during specific slots (defined by \emph{colors}).
We assume a STDMA algorithm that takes
advantage of the regularity of grid topology to also provide a spatially
periodic coloring (``tiling'' of the same color pattern).
In this setting, the key challenges are: 1) minimizing the average routing delay
by ordering the slots in the cycle
2) being energy efficient.
Our work follows two directions: first, the
baseline performance is evaluated when nothing specific is done and the colors are randomly ordered in the STDMA cycle.
Then, we propose a solution, ORCHID 
that deliberately constructs an efficient STDMA schedule. It proceeds in two steps.
In the first step, ORCHID starts form a colored grid and builds a hierarchical routing based on these colors. In the second step, ORCHID builds a color ordering, by considering jointly both routing and scheduling so as to ensure that any node will reach a sink in a single STDMA cycle.
We study the performance of these solutions by means of  simulations and modeling. 
Results show the excellent performance of ORCHID in terms of delays and energy compared to a shortest path routing that uses the delay as a heuristic. We also present the adaptation of ORCHID to general networks under the SINR interference model.

}
\RRmotcle{R\'eseaux de capteurs sans fil, grille, efficacit\'e \'energ\'etique, d\'elais, ordonnancement de l'activit\'e des noeuds, coloriage des graphes, routage, ORCHID, optimisation multi-couche, aggr\'egation, arbre dominant.}

\RRkeyword{wireless sensor networks, grid, energy efficiency, delay, node activity scheduling, graph coloring, routing, ORCHID, cross layer, aggregation, dominating tree.} 
\RRprojet{Hipercom2}  
\RCParis 
\newtheorem{property}{Property}
\newtheorem{notation}{Notation}
\newtheorem{corollary}{Corollary}
\newtheorem{method}{Method}
\newtheorem{lemma}{Lemma}
\newtheorem{theorem}{Theorem}
\newtheorem{remark}{Remark}
\newtheorem{definition}{Definition}

\newcommand{\ZZ}{\mathbb{Z}}
\def\QED{\mbox{\rule[0pt]{1.5ex}{1.5ex}}}
\def\proof{\hspace*{-1.4em}{{\itshape Proof: }}} 
\def\endproof{\hspace*{\fill}~\QED\par\endtrivlist\unskip}

\newcommand{\minipagesize}{\textwidth}
\newcommand{\mybox}[1]{
\vspace{1mm}\fbox{\begin{minipage}{\minipagesize}#1\end{minipage}}\vspace{1mm}}
\newcommand{\tred}[1]{{\color{red}{#1}}}

\newcommand{\REPLACE}[2]{{#2}}
\includecomment{conf}
\excludecomment{rr}
\newcommand{\ForConf}[1]{{}}

\usepackage{titletoc}
\dottedcontents{section}[3.5em]{}{2.2em}{1pc}
\dottedcontents{subsection}[5.5em]{}{3.2em}{1pc}
\dottedcontents{subsubsection}[8.5em]{}{4.2em}{1pc}
\begin{document}

\makeRR   
\tableofcontents
\newpage


\part{General Introduction and System Model}
\section{Introduction}
\subsection{Context}
TDMA scheduling is one of the energy efficient techniques used in WSNs. Each node is allocated a time slot for data transmissions. The main drawback of this pure TDMA is the underutilization of the bandwidth. To limit this problem, many solutions providing bandwidth spatial reuse are proposed (STDMA: Spatial Reuse TDMA))~\cite{STDMA1,STDMA2}. One technique to perform such spatial reuse is the graph coloring. Generally those solutions work as follows. The graph is colored such that no two nodes that have the same color interfere. As an example, let us consider the linear network depicted in Figure~\ref{fig:ntwk}. Assuming that interferences are limited to $2$ hops, the same color is assigned to nodes $A$ and $D$, while a second color is assigned to nodes $B$ and $E$ and a third color is assigned to nodes $C$ and $F$ because these couples of nodes are $3$ hops away. Consequently, they can transmit in the same slot without interfering (see cycle in Figure~\ref{fig:cycle}).
\begin{figure}[!h]
	\centering
		\subfigure[Sample linear network.]{\includegraphics[width=2in]{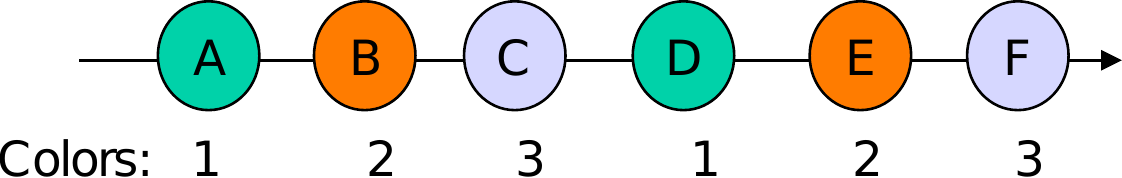}\label{fig:ntwk}}\hspace{11pt}
		\subfigure[STDMA cycle.]{\includegraphics[width=1.7in]{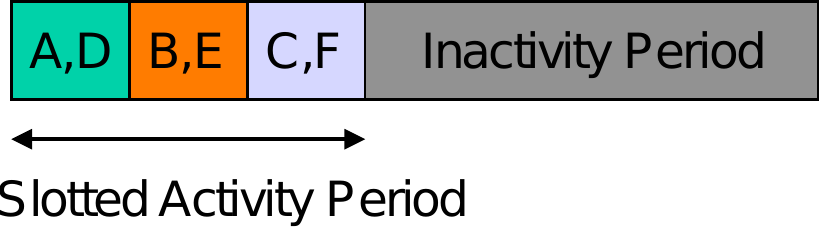}\label{fig:cycle}}
		\caption{Example of a network and the STDMA cycle associated.\label{fig:intro}}
\end{figure}

The slots form the activity period in the STDMA cycle (also called superframe). The remaining part of the STDMA cycle corresponds to the inactivity period where all nodes can turn off their radio and save energy. Consequently, this scheme avoids energy wasted in idle listening and collisions.

\subsection{Sources of Delays in STDMA}
While the main advantage of STDMA in WSNs is energy saving, it is also deterministic. Indeed, contention-based protocols suffer from collisions especially under high data rate. This factor makes the end-to-end delays unpredictable. 
For delay sensitive applications like fire detection, this is not acceptable. 
However, it does not suffice to assign slots to nodes in an arbitrary way to obtain minimum end-to-end delays. Two reasons may increase delays:
\begin{itemize}
\item \textbf{First reason}: Due to the fact that interfering nodes cannot have the same slot, the cycle obtained may be very long. A schedule of minimum length can be achieved by maximizing the reuse of time slots. Therefore, most existing algorithms aim at maximizing the number of transmissions scheduled in the same slot and enable spatial reuse by devising strategies to eliminate interferences~\cite{FASTDATA_Collec,Ergen10}. 

\item \textbf{Second reason}: While routing protocols address the delay issues by choosing the shortest path, the order of medium access by nodes on this path is also crucial. 
Let us consider the example of Figure~\ref{fig:ntwk} and the two schedules depicted in Figures~\ref{fig:badAlloc} and~\ref{fig:goodAlloc}. In the second schedule, colors purple and orange are associated with slots $2$ and $3$ respectively instead of slots $3$ and $2$ in the first schedule. Notice that, consequently, slots assigned to $E$ and $F$ are switched.

\begin{figure}[!h]
\centering
		\subfigure[First ordering: green$\rightarrow$slot1, purple$\rightarrow$slot2, orange$\rightarrow$slot3.]{\includegraphics[width=3.3in]{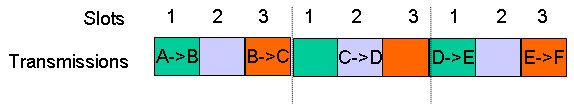}\label{fig:badAlloc}}
		\subfigure[Second ordering: green$\rightarrow$slot1, orange$\rightarrow$slot2, purple$\rightarrow$slot3.]{\includegraphics[width=2.5in]{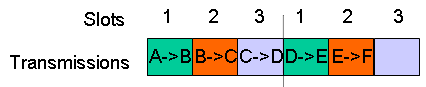}\label{fig:goodAlloc}}
		\caption{Examples of two color ordering for the same network.\label{fig:goodBad}}
\end{figure}
Following the first scheduling (Figure~\ref{fig:badAlloc}), node $A$ requires $9$ time slots to transmit data to node $F$. However, if nodes follow the second schedule (Figure~\ref{fig:goodAlloc}), only $5$ slots are required. 
The reason is that, in the first scheduling, node $B$ receives the packet from $C$ after the end of its slot. However, in the second schedule, slots are sorted according to the same sequence as nodes appear in the route from $A$ to $F$.
\end{itemize}

This result means that each time there is an inversion in the slot ordering compared to the appearance order of nodes in the path, delays are increased by one cycle. 
Consequently, to minimize the delays, in addition to minimizing the cycle length, finding a smart scheduling of the transmissions between any couple of transmitter and receiver is required. This is the main idea of this work. We focus on delay optimization of the scheduling done by a coloring algorithm called VCM~\cite{VCM, rr-VCM}. VCM is dedicated to grid networks and provides a periodic coloring obtained by repeating a specific color pattern. 

\subsection{Contributions}
We treat the aforementioned problem following two directions:

\begin{itemize}
\item \textbf{ First}, we study the performance of our baseline, \textbf{a random} assignment of slots (or colors) in the cycle, and we estimate \textbf{the average normalized delay per radio range}. 
We evaluate this delay by simulation (based on shortest-delay routing and greedy routing), and by a stochastic model (Section~\ref{sec:randomIRCO}).

\item \textbf{Second}, we propose a solution for delay and energy optimization. Our contributions are summarized as:\\
\indent $\circ $ In general, the delay optimization is either managed by routing or by scheduling separately. In our work, we integrate both approaches using cross-layering (i.e. jointly).\\ 
\indent $\circ $ Usually researchers consider that minimizing delays in STDMA is done by minimizing the STDMA cycle length. We believe that, in addition, slot misordering must be investigated. Hence, in our work, we consider both problems.\\
\indent $\circ $ We design a solution called \textbf{ORCHID ("Optimized Routing and sCHeduling in grID wireless sensor networks")} (see Section~\ref{sec:overview}) which includes: 
\begin{enumerate} 
\item A hierarchical routing architecture where sensor nodes route data to a predefined set of data aggregators and these aggregators route data to the data sink.
\item A coloring of nodes and aggregators that allows sensor nodes to collect data to the closest aggregator in a single cycle. 
\item An orchestration of the medium access of sensors and aggregators in the STDMA cycle, that is the colors ordering in the cycle.
\end{enumerate}

\item The produced delays are optimized: ORCHID guarantees data collection in a single cycle.
\item ORCHID supports link failure, multiple sinks and mobile sinks:
indeed, any node is able to reach any aggregator in a single cycle. Hence, this node can route data to any of these aggregators. This also makes ORCHID fault tolerant.
\item Paths between aggregators are also fast: less than or equal to one cycle. 
\item Simulations show that ORCHID outperforms shortest-delay path routing in terms of delays and energy. 
\item \textbf{An adaptation of ORCHID to general graphs under the SINR model} is also presented.
\end{itemize}

\subsection{Organization of the Report}
This report contains $6$ parts. 
\begin{enumerate}
\item The first part is a general introduction describing the context, the problem and the contributions.
Section~\ref{sec:background} presents the background of this work. In Section~\ref{sec:genAssum}, we present the general assumptions and the system model that we consider in this work. 
\item The second part deals with the random ordering of colors in the cycle. We evaluate the performance of the system under this assumption by stochastic model (see Section~\ref{sec:shortModel}) and by simulation (see Section~\ref{sec:simRandom}). 
\item The third part focuses on the proposed solution ORCHID. Section~\ref{sec:overview} gives an overview of the proposed solution and its main principles. Sections~\ref{sec:aggregators}, \ref{sec:ordering}, \ref{sec:highways} and \ref{sec:orchestration} detail the components of ORCHID. 
In Section~\ref{sec:results}, we present the performance evaluation of ORCHID.
\item The fourth part proposes a model (see Section~\ref{sec:model}) for the determination of the distribution of the average delay per range. 
\item The fifth part presents ORCHID for general graphs under the SINR model. 
\item The sixth part concludes the report.
\item The seventh part is an annex providing figures and curves illustrating some obtained results, in addition to those presented in the previous parts.
\end{enumerate}

\section{Background}\label{sec:background}
This section provides a state of the art about (1) Delay efficient TDMA, (2) Geographic routing and (3) VCM: Vector-Based Coloring Method since it is the coloring method that ORCHID uses.


\subsection{Delay Efficient TDMA}
Most delay-aware TDMA scheduling solutions for WSNs attempt to minimize the schedule length~\cite{FASTDATA_Collec, Ergen10}. Other works rather focus on delays induced by the misordering of slots. In general, this latter class includes two categories: the first one assumes that routing is given and optimizes the scheduling accordingly. Inversely, the second one deduces the routing from a given scheduling.
Examples of the first  category include~\cite{cheng2011,cheng2013,chatterje2009,DMAC04,serenaIJDIWC}.
In~\cite{cheng2011}, authors present an integer linear programming (ILP) formulation of the minimum latency link scheduling problem. Assuming that routes are given, the model is able to find the slot used by each link while minimizing the latency. 
We believe that the linear programming approach is not suitable for large networks. 
Solutions~\cite{chatterje2009,DMAC04,serenaIJDIWC} assume that a data gathering tree is given. They ensure that any node has a slot earlier than its parent in the routing tree. Hence only one cycle is required for data aggregation. 
For instance, DMAC~\cite{DMAC04} allocates slots to nodes based on their positions in the data gathering tree. Nodes at the deepest levels in the tree are assigned the earliest slots. In SERENA~\cite{serenaIJDIWC} which is a scheduling algorithm based on graph coloring, each node has a color higher than the color of its parent in the data gathering tree. Consequently, when slots are assigned according to the decreasing order of colors, any parent node has a slot that follows the slots of all its children. We have shown in~\cite{serenaIJDIWC} that the cost of such ordering is the increase of the cycle length. Unlike DMAC, SERENA is conflict free while ensuring the spatial reuse. 


In the second category, 
we can cite Yu et al.~\cite{FAN2007}. Indeed, each link $(i,j)$ is assigned a cost that is equal to the time difference between slots assigned to nodes $i$ and $j$ respectively. The advantage of this solution is that all routes (from any node to another) are possible without scheduling change, but on the other hand nothing makes them energy-efficient. CBS~\cite{Deng13} is based on a tree formed by cluster heads that relay data from other nodes. Both intra and inter-cluster communications are TDMA based. Furthermore, to minimize delays, each node has a slot after its parent in the tree. Also, the authors propose an asynchronous method for cluster heads communications.

\subsection{Geographic Routing}\label{sec:ref-geo-routing}
Geographic routing algorithms use position information to perform  packet forwarding. 
 Most geographic routing algorithms use a greedy strategy that tries to approach the destination in each step, e.g. by selecting the neighbor closest to the destination as a next hop. However, greedy forwarding fails in local minimum situations, i.e. when reaching a node that is closer to the destination than all its neighbors. 
Furthermore, there are other several greedy routing strategies that select the relay based on the distance and the direction towards the destination~\cite{surveyGeoR1}.~\cite{surveyGeoR1} and~\cite{surveyGeoR2} provide a good survey of geographic routing protocols in ad hoc networks. TTDD~\cite{TTDD} and PGP~\cite{PGP} support sink mobility and propose a hierarchical routing structure. This structure is based on a virtual grid which divides the sensor field into many cells. Then the sensors on the grid points will serve as dissemination nodes. PGP optimizes the induced overhead of TTDD. 


Geographic routing has also been proposed as a good approximation to minimize end-to-end delays. The idea~\cite{GERAF} is to carry packets to their destination in as few hops as possible, by making as large progress as possible at each relaying node. 
In this context, related works turn around, among others, the determination of the number of hops of the optimal path~\cite{k-hopStat} or the optimal transmission range as in~\cite{kleinrock} or the optimal delay per hop~\cite{CMAC}. 
Ma et al.~\cite{k-hopStat} proposed a geometric model to compute the conditional probability that a destination node has hop-count $h$ with respect to a source node given that the distance between the source and the destination is $d$. 

Not very far from our context,~\cite{naveen, CMAC, GERAF} proposed geographic routing to minimize delays considering a duty cycle based medium access. The heuristic that~\cite{GERAF, k-hopStat} use is the progress towards the destination.  
Unlike~\cite{GERAF, k-hopStat}, we consider that a good heuristic to compute one-hop optimal delays has to take into account not only the progress towards the destination but also the transmission latency. This is the heuristic that we use in the design of our model for delay per range determination (see Section~\ref{sec:model}). 


\subsection{VCM: Vector-based Coloring Method}\label{sec:vcm}
VCM~\cite{VCM,rr-VCM} is a graph coloring algorithm dedicated to dense grid networks. Indeed, the grid topology is one of the best methods to ensure full area coverage in surveillance applications~\cite{BFKKP10,CIQC02}. From the research point of view, working on grids can be a first step towards general graphs as it is easier especially for theoretical studies. \\
The intuitive idea of VCM is as follows. As the grid topology presents a regularity in terms of node positions, VCM produces a similar regularity in terms of colors and \textbf{generates a color pattern that can be periodically reproduced to color the whole grid}. An example of the coloring provided by VCM is given in Figure~\ref{fig:vcm}.
\begin{figure}[!h]
\centering
\includegraphics[width=0.5\linewidth]{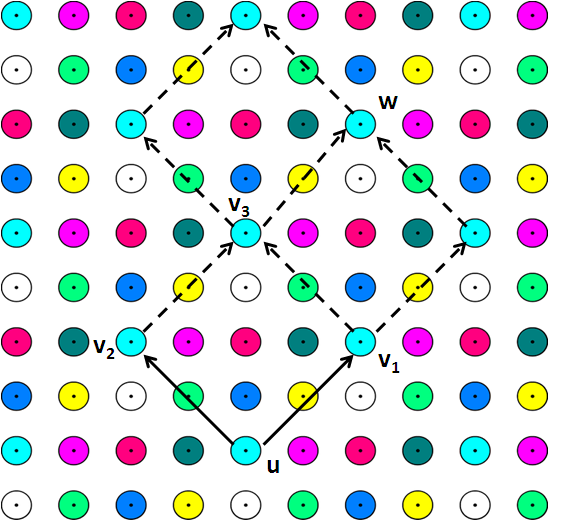}
\caption{An example of $3$-hop coloring provided by VCM (no 1 or 2 or 3-hop neighbors share the same color.)}\label{fig:vcm}
\end{figure}

The color pattern is defined by two vectors ($\vec{uv_1}$ and $\vec{uv_2}$ in Figure~\ref{fig:vcm}) called \textit{VCM generator vectors}. VCM finds a color pattern that minimizes the number of colors used and hence minimizes the length of the TDMA cycle. VCM coloring has the following properties:
\begin{enumerate}
\item Nodes having the same colors form a lattice of generator vectors $\vec{uv_1}$ and $\vec{uv_2}$. Moreover, nodes with the same color are nodes with the same relative coordinates according to the parallelogram they belong to. Consequently, color computation is based on these relative coordinates. 
\item All colors produced by VCM are the colors inside one parallelogram: The number of colors produced by VCM to color the whole grid for any transmission range $R$ is equal to the scalar product of the two generator vectors.
\item The number of colors produced by VCM does not depend on the size of the grid, but on the transmission range of sensor nodes.
\end{enumerate}
In this work, we adopt this algorithm because: \begin{enumerate}
\item Nodes compute their colors using a simple formula without requiring message exchange with neighbors. 
\item VCM uses a protocol-based interference model, but can be extended to SINR model (see Part~\ref{sec:orchid++}).
\item The schedule that VCM provides is optimal in terms of length when the transmission range tends to infinity. Indeed, we have proved in~\cite{rr-VCM}, that the number of colors is given by Property~\ref{theo:nbColor}.
\begin{property}\label{theo:nbColor} 
(From \cite{rr-VCM}) The number of colors $n_c(R) $ of an optimal periodic $h$-hop coloring, including VCM, for a fixed $h$
verifies:
$$n_c(R) = \theta R^2(1 + O(\frac{1}{R})) \mathrm{,~when~} R \rightarrow \infty \mathrm{,~with~} \theta = \frac{\sqrt{3}}{2}h^2.$$
\label{th:asymptotic}
\end{property}

\end{enumerate}

\section{System Model and General Assumptions}\label{sec:genAssum}
For the remaining of this report, we consider the following system model. \\
$\bullet $ \textbf{Network}: Fixed nodes form a grid of a grid step $=1$. The idea is to study our general strategy of coloring (ORCHID), and its performance on such grid networks, to provide insights on the performance of such specific networks. Later in Part~\ref{sec:orchid++}, we will consider extension to general graphs. \\
$\bullet$ \textbf{Sinks:}
Sinks can occupy a predefined set of positions of the grid (described later) and may move from one position to another.  \\
$\bullet $  \textbf{Neighborhood}: As in several wireless network studies, we adopt a unit disk communication model~\cite{gsw1998}. Any node can communicate directly with all nodes within a disk centered at itself with a communication range denoted $R$.\\
$\bullet$ \textbf{Application:} We are motivated by military surveillance applications where sensors monitor a region and transmit periodic measures or alarms to the sink. Long range communications are assumed; multi-hop routing is needed to deliver data to the sink. We suppose that data can be aggregated at any intermediate node and that packets are small. Hence, it becomes realistic to assume that any node is assigned a single time slot and that this time slot suffices for all its transmissions.\\
$\bullet $  \textbf{Transmission and Interference Models}: As in several wireless network studies, we adopt a unit disk communication model~\cite{gsw1998}. Any node can communicate directly with all nodes within a disk centered at itself with a communication range denoted $R$. The interference model is given by the coloring algorithm VCM, by default protocol-based: in Section~\ref{sec:orchid++}, the SINR interference model is considered and we adapt the proposed solution to this assumption.\\
$\bullet $ \textbf{Medium access}: is based on STDMA where each node has a slot to transmit data, and more than one node can share the same slot. These slots are granted based on colors computed by VCM. \\
$\bullet $ \textbf{STDMA cycle inactivity}:
in performance evaluations, for simplicity, we assume a STDMA cycle without inactivity period, and the delay is expressed as the number of slots necessary for a packet
to travel from a source to a destination\footnote{Notice that for a large delay $D$ of spanning several cycles of length $C$, the number of inactivity periods is directly estimated as $\lfloor D/C \rfloor$ with an error of at most $1$. Hence even with inactivity periods, for large $D/C$ mostly the value of the delay $D$ is of interest.} (as in Equation~\ref{eq:delay}).\\
\\
$\bullet $ \textbf{One-Hop Delay}: Let 
$D{(i,j)}$ be the induced one-hop delay for a packet of node $i$ that is further relayed by node $j$. Assume that node $i$ is assigned the slot $s_i$ while $j$ is assigned the slot $s_j$. Hence, $D{(i,j)}$ is equal to the difference between the slot where $i$ starts transmitting  and the slot $s_j$ where the node $j$ relays the last bit of the packet of $i$. Let $S$ be the total number of slots. Like in~\cite{FAN2007} and assuming that there is no inactivity period, we have:
\begin{equation}\label{eq:one-hop-delay}
D{(i,j)} = \left\lbrace
	\begin{array}{rl}%
  		s_j-s_i &\mbox{ if $s_j>s_i$} \\
  		S+s_j-s_i &\mbox{ if $s_j<s_i$}%
      \end{array} \right.%
\end{equation}
\\
$\bullet $ {\bf Route Delay:} on a route (path) $\mathcal{P}$ of successive nodes $(v_1,v_2,v_3, \ldots, v_k)$, the delay
is computed as the sum of successive one-hop delays except for the last node 
(which is the destination and does not relay the packet).
\begin{equation}\label{eq:delay}
D(\mathcal{P}) = D{(v_1,v_2)}+D{(v_2,v_3)}+\ldots+D{(v_{k-2},v_{k-1})}
\end{equation}
\\

\clearpage
\part{Random Color Ordering}
This part of the report assumes an arbitrary assignment of colors in the cycle. The objective is to evaluate the performance of the delay under such assumptions.


\section{Overview and Problem Statement}\label{sec:randomIRCO}\label{sec:IRCO}
In addition to the assumptions presented in Section~\ref{sec:genAssum},  we consider a grid colored with VCM and we assume that \textbf{the colors are randomly associated with slots}. This assumption let us define a basic method for building the global cycle. We denote this method \textbf{IRCO (``Iterated Random Color Ordering'')}. IRCO consists in randomly ordering colors of a single VCM cycle and then repeating successively this \emph{base VCM cycle} several times in a row before the inactivity period. For instance, one IRCO cycle is a repetition of several identical VCM cycles necessary to allow all data from any source to reach its destination. Notice that IRCO concerns colors ordering and is independent from the routing adopted. In the following, we will evaluate IRCO with shortest-path routing and greedy routing. 

The aim of this section is to answer the question: \textit{What is the performance of routing with a random coloring, in terms of delay?} 

We answer this question by designing a stochastic analytical model (brief description in Section~\ref{sec:shortModel} and more details in Section~\ref{sec:model}) and 
by simulation in Section~\ref{sec:simRandom}. Results of the two techniques are also compared at the end of Section~\ref{sec:simRandom}.

\section{Performance Analysis with a Stochastic Model}\label{sec:shortModel}
We design a model to estimate the average normalized delay per range. Unlike previous works (such as~\cite{GERAF, k-hopStat}), this model selects as relay the neighbor yielding the best normalized delay per range. We have proved the following property, in Section~\ref{sec:modelProof}:
\begin{property}\label{prop:eqModel}\label{prop:NDR-IRCO}
An estimate of the normalized delay per range is $\frac{3}{2}\theta + \frac{3}{4}\pi$, when the radio range $R$ grows and the number of colors is $\theta R^2$, where $\theta = \frac{\sqrt{3}}{2}h^2$ (for a h-hop coloring\footnote{a h-hop coloring prevents nodes up to $h$ hops to use the same color.}).
\end{property}
The implication of this result is rather strong: despite that fact that the number of colors grows as $R^2$, with random color ordering, to travel a distance 
of $K$ hops, the delay in ``number of slots'' stays equal to a constant ($\frac{3}{2}\theta + \frac{3}{4}\pi$)
multiplied by $K$.
For example, for a $3$-hop coloring with VCM, the estimate of the normalized delay per range is $\approx 14.05...$ which is rather low.

\section{Performance Analysis with Simulations}\label{sec:simRandom}
Assuming IRCO, the routing method that we use for simulations is one of:

\subsection{Shortest Delay Routing} 

Because the delay as defined in Section~\ref{sec:genAssum} is an additive
metric, it is possible to compute the ``shortest path'' using classical
algorithms (Bellman-Ford, Dijkstra) with this metric, as proposed for instance by~\cite{kim}.
The paths that are found are the shortest paths in terms  of delay (expressed as in Equation~\ref{eq:delay}) necessary to reach the destination. Note that equivalently, they yield the lowest normalized
delay par range computed over the whole route.

\subsection{Greedy Routing Heuristic}\label{sec:local-heur}



One possible strategy for routing on a colored grid, is the following: at every hop, decide what would be the next hop, based only on local information (and not global information). This is a strategy derived from geographic routing (see Section~\ref{sec:ref-geo-routing}).
For delay-aware routing, it is better to select a relay node presenting the best compromise between delay and progress. Such a \emph{greedy} heuristic is proposed, for instance, in CMAC~\cite{CMAC}: the selected next-hop is the one that yields the lowest ratio ``delay''/''progress towards the source''. 
We can adopt this heuristic in this context of grid network and predefined
STDMA. 
Note that this is equivalent to yielding the lowest normalized delay per range computed on one hop.
%
%
%
%
%
\section{Experimental Results}\label{sec:RCO-result}
%
The objective of this section is to provide statistics about the average normalized delays. The parameters of simulation are summarized in  Table~\ref{tab:sim-randomparam}.
The two simulated strategies are the greedy routing and the shortest-delay path denoted respectively ``\emph{Sim. greedy end-to-end}'' and ``\emph{Sim. shortest-delay path end-to-end}'' in Figure~\ref{fig:random-sim}.
For each strategy, source nodes were selected according to Table~\ref{tab:sim-randomparam}. The path from each source to the center point $(0,0)$ was computed. We first compute the  total end-to-end delay (expressed in slots). Then, we obtain the normalized delay per range by dividing this delay by 
the geometric distance (normalized by radio range). An average was taken on all selected nodes, and on multiple simulations.
\begin{table}[!h]%
\centering%
\begin{tabular}{|p{3cm}|p{5cm}|} \hline
Nb simulations & 100 random orderings with same VCM vectors \\ \hline
Radio range & from $1$ to $7$ with step $1/4$ \\ \hline
Routing strategy & shortest-delay path or greedy heuristic \\ \hline
Size of the grid & $L \times L$ with $L = 601$ \\ \hline \hline
Sources for Greedy routing & 100 random nodes at distance $D$ 
with $0.9 L \le D \le L$ \\ \hline
Sources for Shortest-Delay Path routing & all nodes at distance $D$
 with $0.9 L \le D \le L$ \\ \hline
\end{tabular}%
\caption{Simulation parameters\label{tab:sim-randomparam}}%
\end{table}%

For reference, we also computed the asymptotic estimate of the model given by Property~\ref{prop:eqModel} (plot with caption ``\emph{Model}'' in Figure~\ref{fig:random-sim}). 
%
\begin{figure}[!h]%
\centering%
\includegraphics[width=0.75\linewidth]{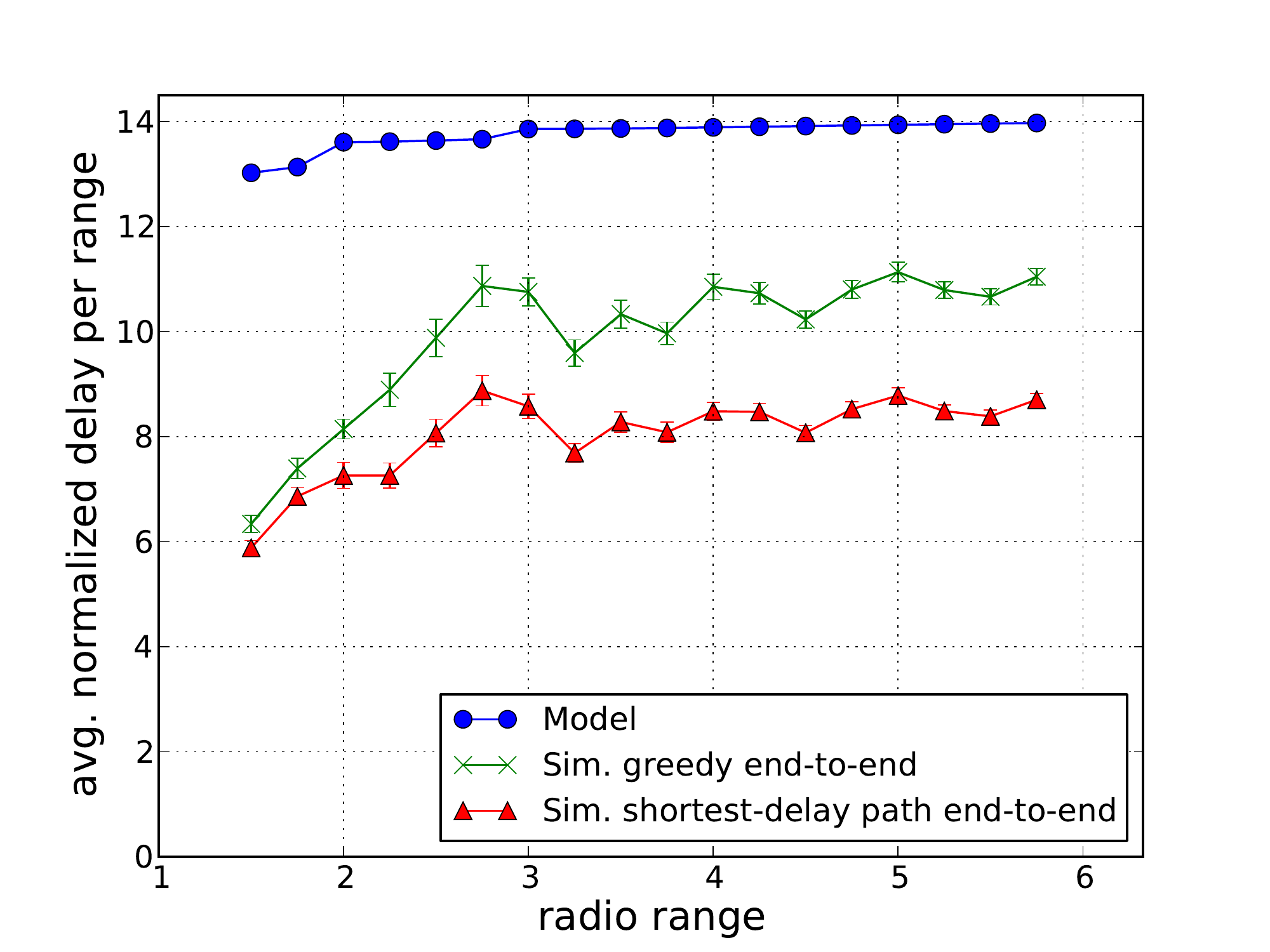}\vspace{-3mm}%
\caption{Normalized delay for random color orderings.}\label{fig:random-sim}%
\end{figure}%
The results are represented in Figure~\ref{fig:random-sim}.\\
The main result is that the normalized delay per range increases only slowly with the radio range (although the number of colors increases as the square of the range) and is quickly plateauing (as predicted by the model in Section~\ref{sec:model}).
This illustrates that the random ordering of colors yields a bounded normalized delay even when the number of colors increases to infinity. In other words, random ordering of colors is an optimal ordering when radio range increases.

Furthermore, for both cases of greedy routing and shortest-delay path,
the normalized delay per range is below $12$ (and the asymptotic estimate
is $14.0..$). Knowing that a loose lower bound is $1$, this shows that
the obtained delay is within an order of magnitude of an optimal STDMA.

Comparing shortest-delay path and greedy routing, we see that shortest-delay
path gives approximately $25\%$ better performance than greedy routing. The conclusion is that greedy routing performs quite well, if we take
into account the fact that it does require only local information. 
Indeed, shortest-delay path, requires, of course, a shortest-path computation
of every destination, which is cumbersome for WSNs, in a scenario where the destination(s) are not known in advance.
Notice here the irregularity of the curves. This is explained by the fact that colors by radio range area vary following the same irregularity. Indeed, when we plot the number of colors divided by the area of the disk of radio range ($\pi R^2$), we also see the same irregularity.

\section{Discussion about IRCO Performance}
The expected performance of IRCO is rather good, as deduced from results. However, random ordering suffers from drawbacks. The dimensioning of the entire cycle necessary to allow all data from any source to reach its destination is critical. It depends on the set of destinations and could be estimated from experimental results seen above. The drawback of IRCO is that either dimensioning is only estimated, potentially losing determinism (guarantees of arrival within one IRCO cycle), or the dimensioning requires pre-computing in advance all routes. In addition, all nodes are potential relays, and therefore should be awake on their neighbors slots. 
In the following, we investigate solutions without these drawbacks.


\clearpage
\part{ORCHID: ``Optimized Routing and sCHeduling in grID wireless sensor networks''}
\section{Overview about ORCHID}\label{sec:overview}
In this section, we will describe the problem statement and give an overview about the principles of the proposed solution ORCHID ("Optimized Routing and sCHeduling in grID wireless sensor networks").

\subsection{Assumptions}\label{sec:assumpProblem}
In addition to the assumptions considered in Section~\ref{sec:genAssum}, we present some other assumptions. We assume a data gathering application where sensors monitor a region and transmit periodic measures or alarms to the sink. We assume also that aggregation is possible: any sensor is able to aggregate data it receives with the data it generates. This assumption is realistic especially for low traffic or when sensor nodes perform some operations like computing the average value or the maximum value of the received data. We assume that each node is assigned a single time slot that suffices for all its transmissions. 

\subsection{Problem Statement}
When the scheduling is based on colors, ordering the slots of a given schedule is equivalent to associating each color with a time slot in the right order (cf. Figure~\ref{fig:intro}).\\

\textit{The problem statement: 
Given a grid colored by the periodic coloring VCM, we want to build routes taking into account these colors and then to order these colors in a schedule ensuring that data reach the sink in one cycle. 
}

We will see in the next sections that solving this problem is achieved by a cross layering between the scheduling and the routing.

\subsection{Discussion about Intuitive Solutions}
Before describing the proposed solution, we explain the different steps that lead us to find the solution that we will propose under the aforementioned assumptions.

\subsubsection{First solution: OPERA~\cite{operaSite}}
We have previously proposed the OPERA software composed of the routing EOLSR~\cite{operaSite} and the scheduling OSERENA~\cite{oserena} adapted to general topologies. EOLSR builds a routing spanning tree and OSERENA colors the tree such that the color of any node is higher than the color of its parent. Colors are ordered in the cycle according to the decreasing order. Consequently, this coloring constraint allows the data aggregation in one cycle. However, for dense and large grid networks, VCM is more optimized than OSERENA in terms of number of colors. 
Moreover, applying a routing algorithm that relies on a spanning tree like EOLSR and integrating it with a periodic coloring like VCM raises some issues as explained in the next section.

\subsubsection{Second solution: Spanning Tree with VCM}
This solution consists in building a spanning tree of the colored grid.
\begin{enumerate}
\item Assume first that we keep a random order of VCM colors on the STDMA cycle. This random order might result in slot inversion and hence additional delays as explained in the scenario of Figure~\ref{fig:goodBad}. Hence, what we should do is to order colors of nodes in the routes.
\item Having a spanning tree that contains all nodes raises the following constraints:
\begin{enumerate}
\item Due to the periodicity of VCM colors, it is impossible that any tree branch contains any color just once for large grids. In this case, delays might exceed one cycle.
\item Let a node $N$ and the route of this node towards the sink via the tree. We order the colors of intermediate nodes in this route such that the color of the first node that appears first in this route is associated to the first slot.
This ensures one STDMA cycle for data transmission for this node $N$. However, we cannot do this for all nodes, for two reasons. First, colors are shared between routes (color repetitions with VCM). Second, nodes are located at different geographic locations relative to the data sink. So, optimizing the order of these colors for some nodes according to one direction might make this order not suitable for other directions. 
\end{enumerate}
\end{enumerate}
 
The previous constraints lead us to define the following principles. 
Instead of building a unique spanning tree including all nodes where it is impossible to have delay optimal routes towards the sink for all nodes, we build several trees. Roots of these trees are particular nodes called aggregators. Then, we order the colors of nodes on these trees. Furthermore, each color does not appear more than once in any tree branch. Any aggregator routes data to another aggregator until data is received by the data sink. In the following two sections, a detailed description of the proposed solution ORCHID is given.

\subsection{Principles of ORCHID}
Figure~\ref{fig:overview} illustrates an example of a grid colored by VCM. $u_1$ and $u_2$ are the generator vectors of VCM.
\begin{enumerate}
\item 
Over this grid, ORCHID builds a hierarchical routing: nodes transmit data to a set of aggregators which are the vertices of VCM pattern, and these aggregators route data to the sinks (such as node $A$ in Figure~\ref{fig:overview}). The sinks should be placed themselves at aggregators.
\item Routing and scheduling must ensure that: 
	\begin{itemize}
	\item Any node reaches the closest aggregator in a single cycle. The path followed by this node is denoted '\textit{Route}'.
	\item The number of nodes that are able to reach more than one aggregator in a single cycle must be maximized. Hence, if a node cannot reach the closest aggregator because of link failure, it can route over another aggregator to reach the data sink. Also, this is useful also when there are multiple sinks. In this case, any node can select the closest aggregator for every sink.
	\item Paths between aggregators should be fast. They are denoted \textit{`Highways'} in the remaining of this report.
		\end{itemize}

\item $Routes$ and $Highways$ are activated successively in the STDMA cycle as depicted in Figure~\ref{routeHighways}. 
\begin{figure}[!h]
\centering
\includegraphics[width=1.8in]{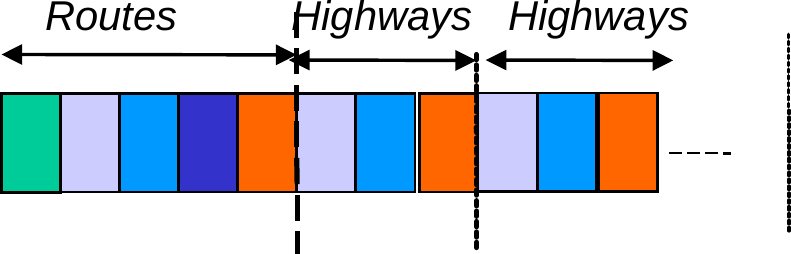}
\caption{Orchestration of $Routes$ and $Highways$ slots in the STDMA cycle.\label{routeHighways}}
\end{figure}
The first activity period of the STDMA cycle is dedicated to $Routes$: activation of VCM colors. Each VCM color occurs only once. 
Unlike the first period which appears only once in the global STDMA cycle, the second period corresponds to the repetition of the $Highway$ cycle a number of times sufficient to allow data sent by aggregators to reach the sink. We will see later that this cycle is divided into sub-cycles. Notice that the number of scheduled colors in this period is smaller than in the first one.
\item The scheduling is based on VCM colors. More precisely, no new colors are used and the spatial reuse obtained by VCM must be guaranteed. Furthermore, these colors must be ordered in the STDMA cycle for both $Routes$ and $Highways$ periods, in such a way that the end-to-end delays are minimized.

\item Any node is active only during the slots associated with its color and the colors of its potential transmitters and can sleep the remaining time, in both $Routes$ and $Highways$ periods.
\end{enumerate}
\subsection{Overview and Steps of ORCHID}
This section provides a more detailed description of ORCHID. Figure~\ref{fig:steps} illustrates its different steps and components.

\begin{figure}[H]
\centering
\includegraphics[width=3.5in]{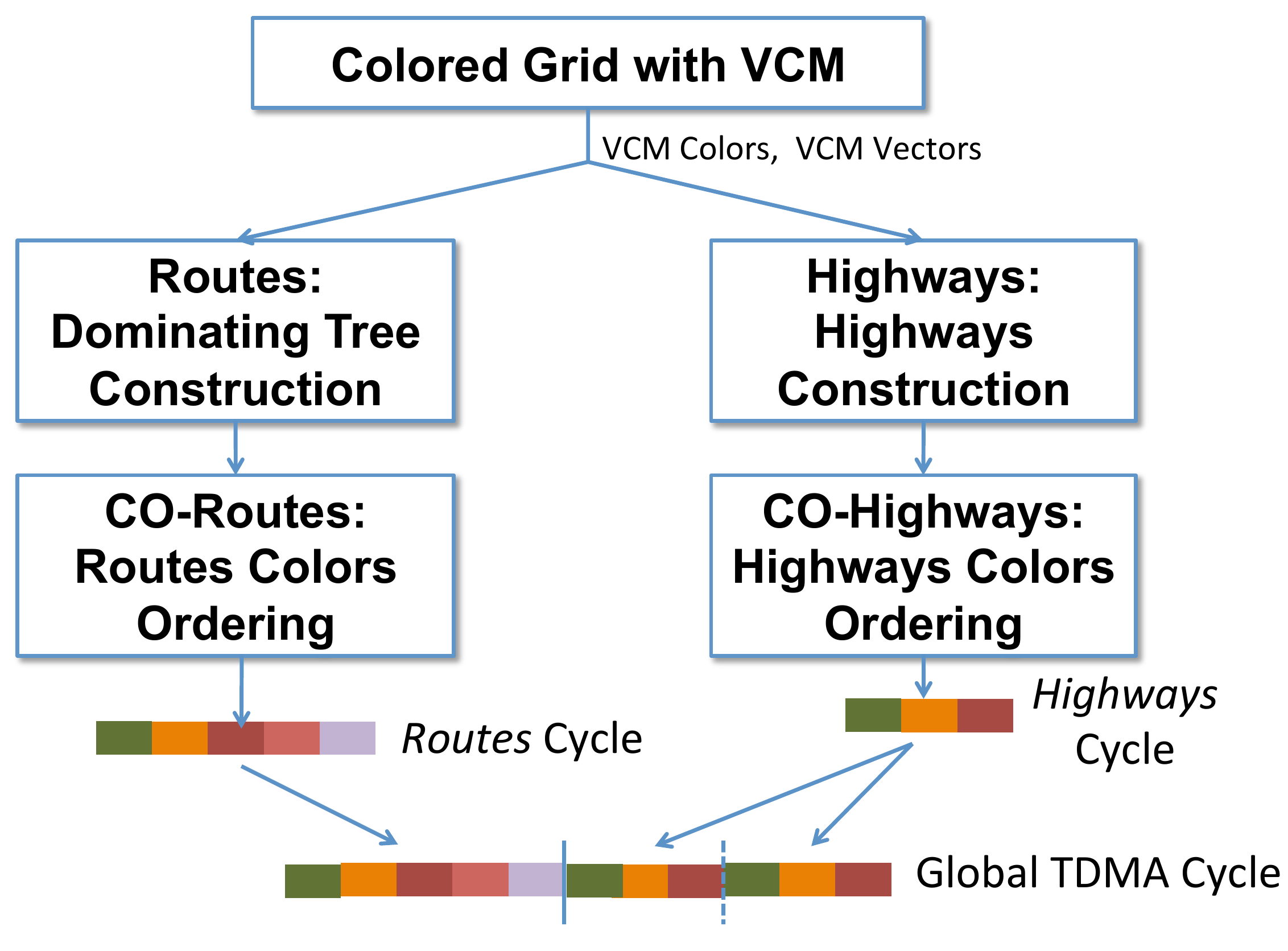}
\caption{Steps of ORCHID.\label{fig:steps}}
\end{figure}
First, the solution is applied to a colored grid of sensors where each node has a color computed by VCM. The principles of the solution can be presented in the following components:

\begin{enumerate}
\item \textbf{Hierarchical routing structure}:
Figure~\ref{fig:overview} illustrates the hierarchical routing proposed. It is composed of $Routes$ (dominating trees) and $Highways$. $u_1$ and $u_2$ are the VCM generator vectors.
\begin{figure}[H]
\centering
\REPLACE{\includegraphics[width=3.2in]{overviewNew}}{\includegraphics[width=3.5in]{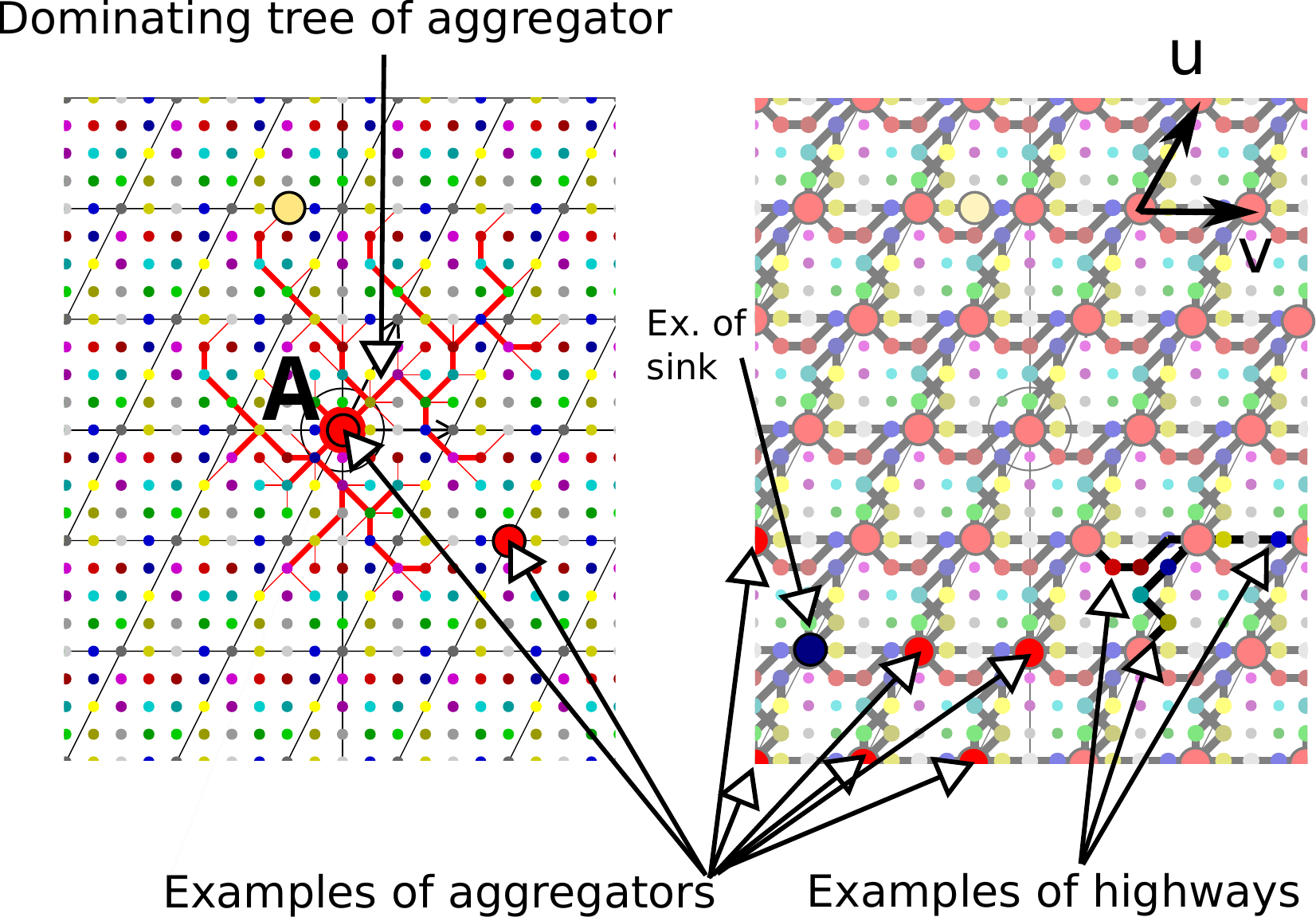}}
\caption{Hierarchical routing structure (left: routing to one aggregator on its dominating tree -- right: mesh routing between aggregators).\label{fig:overview}}
\end{figure}
\begin{enumerate}
\item \textbf{\textit{Routes}}:
\begin{itemize}

\item A set of aggregators is chosen: the vertices of VCM patterns (see Figure~\ref{fig:overview} as an example).
\item Each aggregator $A$ is the root of a dominating tree $\mathcal{T}_{A}$. A dominating tree is defined by Definition~\ref{def:domTree}.
\begin{definition} \label{def:domTree}
A tree $\mathcal{T}$ dominates a set of nodes if any node of this set is either in the tree or has at least one neighbor belonging to this tree. Such a neighbor is called $dominator$.
\end{definition}

In other words, we build a connected dominated set (CDS) constructed as a tree which is rooted at the aggregator. The tree rooted at any aggregator $A_2$ is a copy of the tree rooted at another aggregator $A_1$ after translation by the vector $\vec{A_1A_2}$. 

\item For any aggregator $A$, any tree $\mathcal{T}_{A}$ is as large as possible as long as any color appears only once in this tree branch. 
\item Any node in the grid is necessarily dominated by the tree rooted at the closest aggregator and possibly by trees rooted at other aggregators. As previously said, these other trees are copies of the tree rooted at the closest aggregator after being translated.
\end{itemize}

\item \textbf{\textit{Highways}}: are used by aggregators to join the data sink.
Except aggregators in the border of the grid, any other aggregator has 
$4$ neighboring aggregators (in the grid formed by aggregators). We build $4$ shortest paths between this aggregator and each of these $4$ neighboring aggregators. These shortest paths are called $Highways$.

\begin{remark}
$Routes$ and $Highways$ construction takes into account the colors of the nodes and is based on the same algorithm for all nodes. Hence, the obtained trees and $Highways$ are the same for all aggregators (as illustrated in Figure~\ref{fig:overview}). For this reason, in the following, we only reason about one aggregator.
\end{remark}

\item \textbf{Routing}: 
For \textit{routing}, any node should maintain the list and positions of the aggregators that it can reach in a single cycle. Given a sink with a known position, this node determines the closest aggregator to this sink and routes data over the dominating tree rooted at this aggregator. Hence, it transmits data either to its parent if it belongs to the tree, or its dominator otherwise.
Then, we use geographic routing to determine for each aggregator the following aggregator over $Highways$.
\end{enumerate}
\item \textbf{Ordering colors}:
Recall that the slot assignment is based on colors computed by VCM. To ensure minimized delays avoiding slot misordering, ORCHID performs \textbf{Color Ordering} of $Routes$ and $Highways$. The Color Ordering is defined as:
\begin{definition}[Color Ordering]
Given a set of colored nodes in a given path, each node color is associated with a time slot in such a way that the color of the node that appears first in the path (the farthest node from the destination) is associated with the first slot. 
\end{definition}
Color Ordering ensures that nodes along a path are activated in the STDMA cycle in the same order of their appearance in this path. Hence, slot misordering delays are avoided and data are delivered within the same cycle.
\textit{Color Ordering} of $Routes$ and $Highways$ consists in:
\begin{enumerate}
\item \textbf{\textit{CO-Routes (Color Ordering of $Routes$)}}: 
The color of the first node linked to the tree is associated with slot $1$, the color of the second node is associated with color $2$, etc. Then, colors of remaining nodes are associated with remaining slots.
This method ensures that the slot granted to a node increases with the node level in the tree. Also, any dominated node has a slot after its $dominator$.

\item \textbf{\textit{CO-Highways (Color Ordering of $Highways$)}}: Similarly, for $Highways$ the order of colors is the same as the order of the nodes on the routes from one aggregator to the next one. 
\end{enumerate}

\item \textbf{Orchestration}:
The global STDMA cycle is composed of the $Routes$ cycle and the $Highways$ cycle. Notice that the $Highways$ cycle is smaller. The $Highways$ cycle itself is composed of sub-cycles repeated a number of times sufficient to allow the farthest aggregator to reach the data sink.

\end{enumerate}

In the following, we will detail the steps of ORCHID.

\section{$Routes$: Dominating Tree Construction}\label{sec:aggregators}
Aggregators are nodes with coordinates $\alpha u+\beta v$ where $u$ and $v$ are VCM generator vectors and $\alpha$ and $\beta \in \ZZ$. Each aggregator is the root of a dominating tree.
Hence, ORCHID deals with the dominating tree construction problem. This problem~\cite{CDSExample1, CDSExample2} is similar to the CDS (Connected dominating Set) problem~\cite{CDSSurvey} but the obtained structure must be a tree. Both of these problems are NP-complete~\cite{CDSSurvey, CDSExample1}. The dominating tree problem is defined in~\cite{CDSExample1} as follows: given a graph $G$ with weighted edges, DomTree problem asks to find a minimum total edge weight tree $T$ such that each vertex is either in $T$ or has a neighbor in $T$. In our case, the objective of ORCHID is to connect as many as possible sensors to their closest aggregator. Furthermore, constraints concern essentially colors and behind them the induced delays. More specifically, we define the problem of \textbf{\textit{Dominating Tree Construction}} denoted \textbf{\textit{DomTree}} as: \\
~~\\
\mybox{
\textit{\textbf{DomTree Problem}: Finding a dominating tree with:
\begin{itemize}
\item The objectives:
\begin{enumerate}
\item[\textbf{O1.}] Each node dominated by a tree must reach the root of this tree in one cycle. 
\item[\textbf{O2.}] The number of aggregators that any node can reach in a single cycle is maximized. 
\item[\textbf{O3.}] Any tree dominates the largest possible number of nodes in the grid.
\end{enumerate}
\item The constraints:
\begin{enumerate}
\item[\textbf{C1.}] Each color appears only once in any branch of the dominating tree. 
\item[\textbf{C2.}] Between any two branches of the tree, any two colors appear always in the same order when going on a path towards the root.
\end{enumerate}
\end{itemize}
}
}
The constraints \textbf{C1.} and \textbf{C2}. make the problem different
from (and harder than) a classical connected dominating set/dominating tree 
problem: in fact, they force the tree to be bounded (unlike any 
classical CDS). \\


In this section, we describe our algorithm to answer this problem. 
We start by introducing the following notation where $u$ and $v$ are VCM generator vectors.
\begin{notation}[Parallelogram $\mathcal{P}_{A}$]
By abuse of notation, for a vector $w$ of coordinates $x_w,y_w$, we denote $\frac{w}{2}$ the vector of coordinates $\lfloor \frac{x_w}{2}\rfloor$ and $\lfloor \frac{y_w}{2}\rfloor$\footnote{$\lfloor . \rfloor$ is the symbol for integer part.}.
Then, we define for any aggregator $A$ of coordinates $(0,0)$  the parallelogram $\mathcal{P}_{A}$ with vertices $(0,u,v,u+v)$ translated by $-(\frac{u}{2}+\frac{v}{2})$. 
\end{notation}


\subsection{Principles of the Construction of a Dominating Tree}
We consider the aggregator with coordinates $(0,0)$. The obtained tree will be repeated for all other aggregators (see Remark~1). We propose a connected dominating tree construction algorithm based on a ``greedy" heuristic, as follows:

\begin{enumerate}
\item The tree is built in $2$ steps: for the aggregator $A$ of coordinates  $(0,0)$, we first build a connected dominating tree in $\mathcal{P}_{A}$ only. Then, we extend the obtained tree via the domination of nodes outside this area. The same algorithm is used for the interior and the exterior of this area.
\item The first node added to the tree is its root. All its neighbors are then dominated. 
\item More generally, when any node is added to the tree, all its undominated neighbors are marked as dominated. The algorithm will not attempt to dominate them later.
\item Nodes are dominated in the order of their increasing number of hops towards the root. For two nodes that are at the same number of hops from the root, the node with the smallest geographic distance to the root is dominated first.
\item Nodes belonging to $\mathcal{P}_{A}$ must have dominators inside $\mathcal{P}_{A}$.

\item \textbf{Nodes having VCM colors already used in the tree (color is identical to the color of a node already in the tree) cannot be dominated and cannot be used as dominators.}


\item \textbf{When any node $N$ is linked to the tree after having dominated one node, other nodes that have the same color and that are already dominated, are also linked to the tree.}

\item Selection of dominators is based on a priority denoted $priority$. To define this priority, we use a heuristic. The priority of any node is the total number of the undominated neighbors that will be dominated if this node is added to the tree (with all other candidates of the same color).
\item Steps {\bf 6.} and {\bf 7.} are the greatest difference with classical CDS.
\end{enumerate}

 
\subsection{Algorithm of the Construction of a Dominating Tree}
Given the aggregator $A$ of coordinates $(0,0)$, Algorithms~\ref{algo:domTree},~\ref{algo:addNode}  and~\ref{algo:domNode} illustrate the construction of the dominating tree $\mathcal{T}_A$ rooted at $A$.
\begin{algorithm}
\caption{\textsc{Dominating-Tree-Construction}}\label{algo:domTree}
\begin{algorithmic}[1]
\begin{scriptsize}
\STATE Input: the grid $G$, the aggregator $A$ and the $priority$ (the heuristic for the selection of the $dominator$)
\STATE Output: the dominating tree $\mathcal{T}$  
\STATE $dominatedNodes$ = $\{A\}$
\STATE $treeColors$ = $\{color~of~A\}$ 
\STATE Add all the neighbors of $A$ to $dominatedNodes$
\STATE $insideNodes$ = list of nodes in $\mathcal{P}_{A}$ sorted according to the increasing distance in number of hops relative to the root of the tree. In case of equality, consider the node with the smallest geographic distance to the sink.
\STATE $inside=True$
\FORALL {($toBeDominated$ $\in$ $insideNodes$) /*Dominate nodes inside $\mathcal{P}_{A}$*/
}
	\STATE $dominator$ = \textsc{Dominate-Node} ($\mathcal{T}$, $toBeDominated$, $inside$,  $dominatedNodes$, $treeColors$, $priority$)
	\IF {($dominator \neq NULL$)} 
	\STATE \textsc{Add-Nodes} ($dominator$, $\mathcal{T}$, $dominatedNodes$, $treeColors$) /*Adding dominated nodes with the same color as $dominator$*/
	\ENDIF
\ENDFOR
\STATE $outsideNodes$ = list of nodes $\notin$ $\mathcal{P}_{A}$ sorted according to the increasing distance in number of hops relative to the root of the tree. In case of equality, consider the node with the smallest geographic distance to the sink.

\STATE $inside=False$ 

\FORALL {($toBeDominated$ $\in$ $outsideNodes$)/*Dominate nodes outside $\mathcal{P}_{A}$*/}
	\STATE $dominator$ = \textsc{Dominate-Node} ($\mathcal{T}$, $toBeDominated$, $inside$, $dominatedNodes$, $treeColors$, $priority$)
	\IF {($dominator \neq NULL$)} 
	\STATE \textsc{Add-Nodes} ($dominator$, $\mathcal{T}$, $dominatedNodes$, $treeColors$) /*Adding dominated nodes with the same color as $dominator$*/
	\ENDIF
\ENDFOR
\end{scriptsize}
\end{algorithmic}
\end{algorithm}
\vspace{-11pt}
\begin{algorithm}
\caption{\textsc{Add-Nodes} ($dominator$, $\mathcal{T}$, $dominatedNodes$, $treeColors$)}\label{algo:addNode}
\begin{algorithmic}[1]
\begin{scriptsize}
\FORALL {($node$ $\in$ the grid $G$ with the same color as $dominator$)}
	\IF {($node$ $\in$ $dominatedNodes$)}
	\STATE hasDominated = False	
	\FORALL {($neigh$ neighbor of $node$)}
		\IF {($neigh \notin dominatedNodes$) and (color of $neigh \notin treeColors$)} 
		\STATE hasDominated = True		
		\STATE Add $neigh$ to $dominatedNodes$ 		
		\STATE The dominator of $neigh$ = $node$ 		
		\ENDIF	
	\ENDFOR
	\ENDIF	
	\IF {(hasDominated == True)}
		\STATE Add $node$ to $\mathcal{T}$ \label{line:addT2}
		\STATE The parent of $node$ = the dominator of $node$
	\ENDIF	
\ENDFOR
\end{scriptsize}
\end{algorithmic}
\end{algorithm}


\begin{algorithm}
\caption{\textsc{Dominate-Node} ($\mathcal{T}$, $toBeDominated$,  $inside$, $dominatedNodes$, $treeColors$, $priority$)}\label{algo:domNode}
\begin{algorithmic}[1]
\begin{scriptsize}

\IF {($toBeDominated$ $\in$ $\mathcal{T}$) OR ($toBeDominated$ $\in$ $dominatedNodes$) OR (color of $toBeDominated$ $\in$ $treeColors$) } \label{algo-line:nodeom1}
\STATE return NULL /*This node cannot be dominated*/ \label{algo-line:nodeom2}
\ENDIF
\STATE /*Try to dominate node $toBeDominated$ */
\FORALL {($neighbor \in$ neighbors of node $toBeDominated$)}
	\IF {($inside == True$ and $neighbor \notin \mathcal{P}_{A}$)}
	\STATE continue /*next neighbor*/	
	\ENDIF
	\IF {(($neighbor$ $\notin$ $\mathcal{T}$) OR ($neighbor$ $\notin$ $dominatedNodes$) OR (color of $neighbor$ $\in$ $treeColors$))}\label{line:noDom}
	\STATE continue \label{line:continue}/*next neighbor*/	
    \ELSE 
    \STATE {Add $neighbor$ to $possibleDominators$} \label{line:add}
    \ENDIF 
\ENDFOR 
\IF {($possibleDominators$ is empty)}
\STATE return NULL /*Cannot dominate this node*/
\ENDIF 
\STATE Sort $possibleDominators$ according to $priority$
\STATE $dominator$ = first node sorted in the list $possibleDominators$
\STATE Add $dominator$ to $\mathcal{T}$ \label{line:addT}
\STATE The parent of $dominator$ = the dominator of $dominator$
\STATE Add color of $dominator$ to $treeColors$
\STATE The dominator of $toBeDominated$ = $dominator$
\STATE return $dominator$
\end{scriptsize}
\end{algorithmic}
\end{algorithm}


\newpage
\subsection{Properties of a Dominating Tree}\label{sec:propoDomTree}
In this section, we present some results concerning the dominating tree. 

\subsubsection{General Properties}
\begin{theorem}
Algorithm~\ref{algo:domTree} generates a tree.
\end{theorem}

\proof
When any node is dominated, a new link is added to the tree between this node and its dominator. Consequently, there is no disconnected parts in the grid. Hence the result.
\endproof

\begin{lemma}\label{lemma:domNeigh}
At each iteration, the closest undominated node $N$ to the root $A$ has necessarily at least one neighbor that is dominated by $\mathcal{T}_{A}$.
\end{lemma}

\proof
Assume by contradiction that $N$ has no dominated neighbors. Since nodes form a grid (a connected graph), this assumption means that there exists a shortest path $N, N_1, \ldots, N_k, A$ between $N$ and the aggregator $A$, where $k$ is a positive integer. 
Two cases are possible: 
\begin{itemize}
\item either $N_1$ is dominated, then there is a contradiction because we could select it as a dominator for $N$;
\item  either $N_1$ is not dominated and it is then the closest node to the root and not $N$ as assumed. This results in a contradiction. Hence the proof.
\end{itemize}

\endproof

\begin{lemma}\label{lemma:undominated}
At any iteration, the algorithm attempts to dominate the closest node $N$ to the root aggregator that is not dominated. Two cases are possible,
\begin{itemize}
\item Case 1: $N$ cannot be dominated because:
	\begin{itemize}
	\item its color is used by another node already in the tree;
	\item all of its dominated neighbors have colors already in the tree;
	\end{itemize}

\item Case 2: $N$ can be dominated and a possible dominator exists.
\end{itemize}
\end{lemma}

\proof
First, at any iteration, Algorithm~\ref{algo:domTree} considers the closest node $N$ to the root and attempts to dominate this node if and only if this node is not dominated and this node has not a color used by a node belonging to the tree (see lines~\ref{algo-line:nodeom1} and~\ref{algo-line:nodeom2} in Algorithm~\ref{algo:domNode}).
Second, from Lemma~\ref{lemma:domNeigh}, $N$ has at least one neighbor that is dominated. If this neighbor has a color that appears in the tree, it is not kept (see line~\ref{line:continue} in Algorithm~\ref{algo:domNode}). Otherwise, it is a possible dominator for $N$ (see line~\ref{line:add} in Algorithm~\ref{algo:domNode}). Hence the Lemma.
\endproof


\begin{theorem}\label{theo:trueAlgo}
Algorithm~\ref{algo:domTree} meets the constraints of \textit{DomTree} problem.
\end{theorem}

\proof
According to Algorithms~\ref{algo:domNode} and~\ref{algo:addNode}, any color is linked to the tree in two occasions:
\begin{itemize}
\item First, to dominate any node, its dominator is added to the tree only if its color does not already appear in the tree (see lines~\ref{line:noDom} and~\ref{line:addT} in Algorithm~\ref{algo:domNode}). Hence, no color repetition is produced. Let $dominator$ denotes this dominator.
\item Second, once $dominator$ is linked to the tree, other nodes sharing the same color are added to the tree at the same time if they are already dominated and if they have undominated neighbors (see line~\ref{line:addT2} in Algorithm~\ref{algo:addNode}). 
This is the only time when nodes with this color are added to the tree. Because the tree is "grown" from the root, the same color cannot be repeated twice in the branch.
\end{itemize}
Now, to prove that Constraint \textbf{C2.} is met, let us assume that two colors $c_1$ and $c_2$ are on the same branch. Assume that $c_1$ is closer to the root than $c_2$. This means that all nodes with color $c_1$ had been added in the tree before all nodes with the color $c_2$ and are all closer to the root. Hence, the ordering of colors on the branches must be identical.
\endproof

\begin{theorem}\label{theo:cannotDom}
A dominating tree addressing the $DomTree$ problem cannot dominate all nodes for all sizes of the grids.
\end{theorem}

\proof
From Lemma~\ref{lemma:undominated}, dominating one node becomes not possible when the color unicity of the tree might be not guaranteed. However, the number of colors is finite. Consequently, to guarantee the color unicity in the tree, some nodes might not be dominated. 
\endproof

\subsubsection{Properties Concerning Nodes in $\mathcal{P}_{A}$}

\begin{lemma}\label{lemma:diffColors}
Nodes in $\mathcal{P_A}$ have all different colors.
\end{lemma}

\proof
$\mathcal{P}_{A}$ is nothing other than the VCM pattern translated by the vector $\frac{u}{2}+\frac{v}{2}$. By definition, inside VCM pattern, there is no color reuse.
\endproof

\begin{theorem}\label{theo:dom}
Any node among $\mathcal{P_A}$ is necessarily dominated by the dominating tree rooted at $A$.
\end{theorem}

\proof
By contradiction, let us assume that there is one or several nodes in $\mathcal{P}_{A}$ that are not dominated. Let $M$ be the undominated node that is the closest to the aggregator. 
According to Lemma~\ref{lemma:domNeigh}, node $M$ has at least one neighbor $D$ that is dominated. Because all nodes in $\mathcal {P_A}$ (hence in the dominating tree $\mathcal{T}_{A}$) have different colors (Lemma~\ref{lemma:diffColors}), $D$ could have been selected as a dominator for $M$. Hence the contradiction. 
\endproof

\begin{corollary}\label{coro:closest}
Any grid node is necessarily dominated by the tree rooted at the closest aggregator and possibly by trees rooted at other aggregators. 
\end{corollary}
\proof
The result is true because of Theorem~\ref{theo:dom} and because by construction, an aggregator $A$ is the closest aggregator to nodes in $\mathcal{P_A}$. In addition, a tree    
rooted at $A_i$ can dominate nodes belonging to some $\mathcal{P}_{A_j}$ for $i$ and $j$ different positive integers.
\endproof

\subsection{Results: Examples of Dominating Trees}
Figure~\ref{fig:tree} depicts the dominating tree rooted at the aggregator of coordinates $(0,0)$. The radio range is set to $3\times $the grid step. Thick lines present the tree links while the thin lines are the links between a node and the nodes it dominates. We observe that the obtained tree dominates a larger area: almost $16$ color patterns.
\begin{figure}[!h]
\centering
\includegraphics[width=3.in]{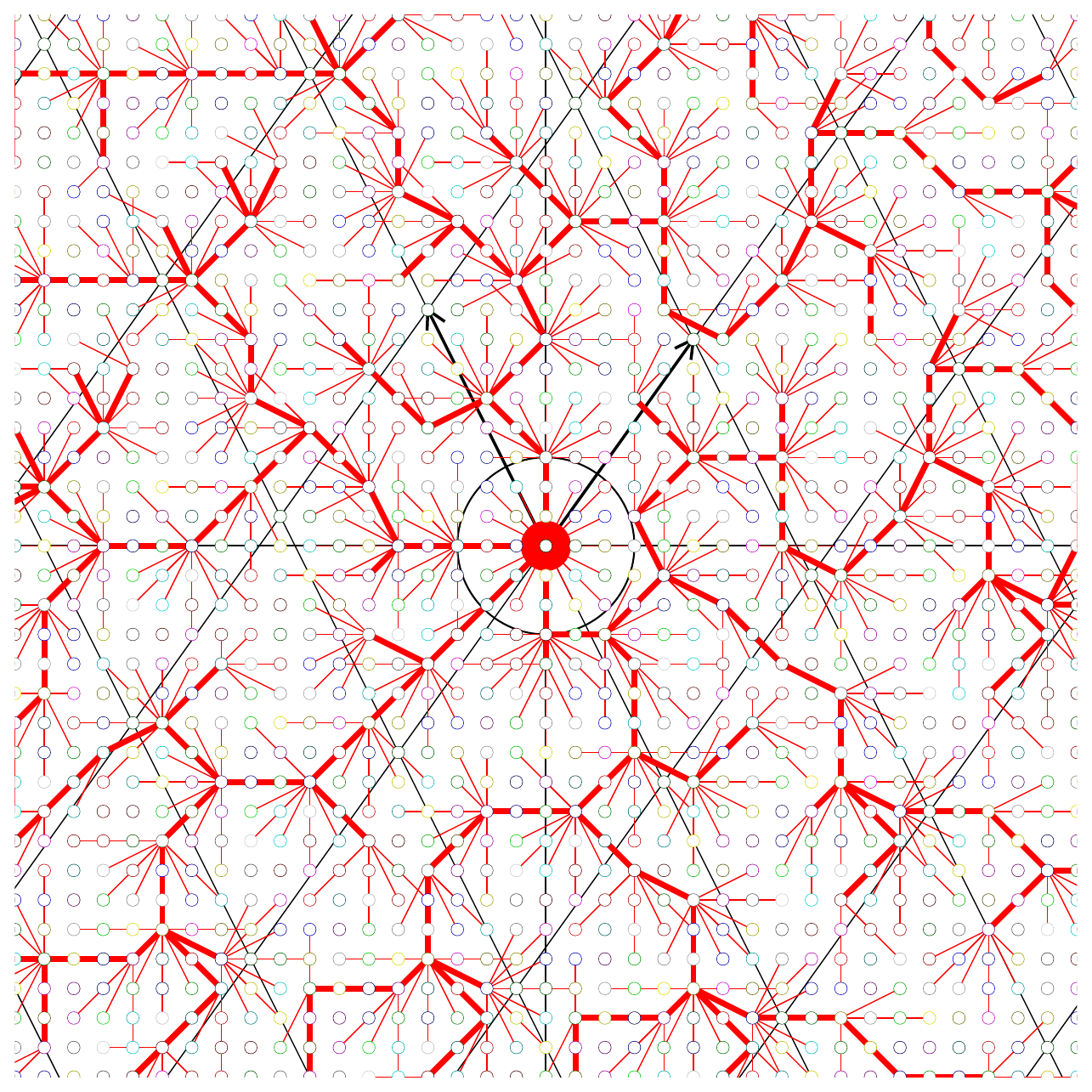}
\caption{The dominating tree obtained for $3$ hop coloring and radio range $=3 \times$ the grid step.\label{fig:tree}}
\end{figure}
~~ \\
Figure~\ref{fig:twoTrees} depicts two trees rooted at aggregators of coordinates $(0,0)$ and $(4,3)$. 
\begin{figure}[!h]
\centering
{\includegraphics[width=3.in]{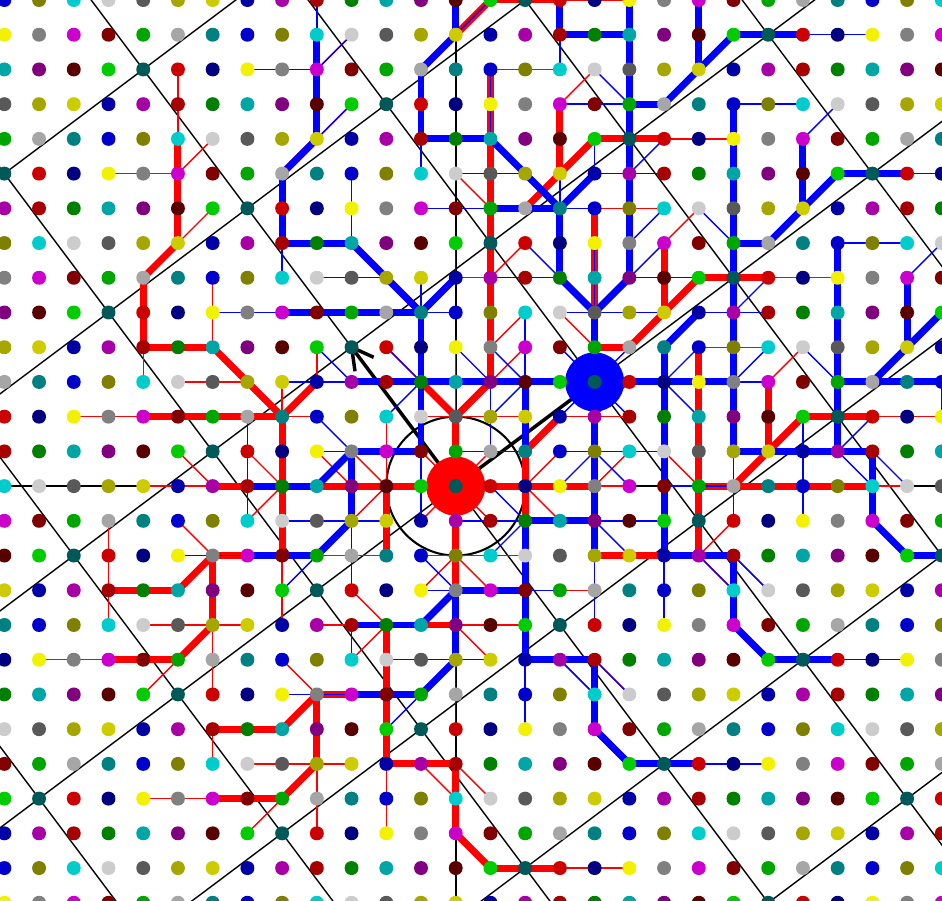}}%
\caption{Two dominating trees for $3$ hop coloring and radio range $=2 \times$ the grid step.\label{fig:twoTrees}}
\end{figure}

Notice that the built trees cover a large area of the grid which answers the objective \textbf{\textit{O3}} of DomTree problem.

\subsection{Routing Over the Dominating Tree}
Any node that is dominated by a tree, including the tree rooted at the closest aggregator, can reach the root of this tree. Its next hop to reach this aggregator is:
\begin{itemize}
\item Either its parent in the tree, if it belongs to the tree.
\item Or its dominator relative to this tree, otherwise.
\end{itemize}

\section{CO-Routes: Color Ordering for $Routes$}\label{sec:ordering}
Our objective is to have \textit{Color Ordering} of $Routes$. 

\subsection{Method of CO-Routes}\label{sec:methOrder}
\begin{method}[CO-Routes]\label{meth:orderColor}

The method is based on the following steps:
\begin{enumerate}
\item \textit{Step1: Nodes of a tree}:\\
Consider the tree $\mathcal{T}_A$ rooted at the aggregator $A$ of coordinates $(0,0)$. We first start by ordering colors of nodes belonging to $\mathcal{T}_A$.
During the construction procedure of the dominating tree (line~\ref{line:addT} in Algorithm~\ref{algo:domNode}), the color of the first node added to the tree is associated with slot $1$, the second (node or set of nodes) is associated with slot $2$, etc. With this method, nodes with deepest levels in the tree are scheduled in the latest slots of the schedule. 
\item \textit{Step2: Remaining nodes}: \\
The remaining grid nodes are sorted according to the increasing number of hops to the root of $\mathcal{T}_A$ and in case of equality the geographic distance is considered as a second criteria.
Then, the colors of these nodes are associated with time slots starting with the last slot of the first step.
\end{enumerate}
\end{method}

\begin{remark}
Method~\ref{meth:orderColor} only orders colors. Hence,  no new colors than those of VCM are used and no color conflict are created.
\end{remark}

\begin{lemma}\label{lemma:insideFirst}
According to Method~\ref{meth:orderColor}, for any aggregator $A$, colors of nodes inside $\mathcal{P}_A$ are ordered before nodes outside.
\end{lemma}
\proof
Colors of nodes linked first to the tree are ordered firstly. Since the dominating tree construction starts by dominating nodes inside $\mathcal{P}_A$, we have the result.
\endproof
\subsection{Results about the Delays Obtained}
\begin{theorem}\label{theo:treeNode}
Nodes belonging to any tree $\mathcal{T}_{A}$ reach $A$ in one cycle.
\end{theorem}

\proof
Let $N$ be any node in $\mathcal{T}_{A}$. By Method~\ref{meth:orderColor}, $N$  will be scheduled after all its ascendants since they are linked to the tree before node $N$. Moreover, from Theorem~\ref{theo:trueAlgo}, no two nodes have the same color among these ascendants. 
Because of Constraint~C2., the ordering of colors is consistent on all branches. Hence, no slot misordering is produced and no more than one cycle is needed by any node $N$ to reach this aggregator $A$. 
\endproof
Theorem~\ref{theo:treeNode} is essential. Indeed, generally the whole activity period in the STDMA cycle is followed by an inactivity period. Hence, delivering data in a single cycle avoids the latency caused by this inactivity period.

\begin{lemma}
All nodes that reach any aggregator $A$ in one cycle are necessarily dominated by the tree $\mathcal{T}_A$.
\end{lemma}

\proof
Assume that a node can reach the aggregator in one cycle. This means that there is a path to the root with only increasing slots numbers. Algorithm~\ref{algo:domTree} adds nodes to the tree one after the other, and orders their colors at the same time.
Iteratively, each node of the path, in the reverse order, will be successively dominated by the next-hop and later linked to the tree. Therefore, the path will be (part of) a branch of the tree. 
\endproof



\begin{lemma}\label{lemma:otherTree}
For any aggregator $A$, for any node $N$ $\in \mathcal{P_A}$, 
$N$ $\notin \mathcal{T}_A$ and $N \in \mathcal{T}_{A_1}$ where $A_1$ is another aggregator, then $N$ is scheduled after its dominator relative to the tree rooted at $A$.
\end{lemma}

\proof
According to Lemma~\ref{lemma:insideFirst}, the colors of nodes inside $\mathcal{P_A}$ is ordered before the part outside. Thus, the color of $N$ is ordered after all nodes in $\mathcal{P}_{A_1}$. This means that the color of $N$ is ordered after all nodes in $\mathcal{P}_A$ (because $\mathcal{T}_A$ is a copy of $\mathcal{T}_{A_1}$). Hence, $N$ has a slot after its dominator relatively to $\mathcal{T}_A$.

\endproof

\begin{theorem}\label{theo:oneCycleTree}
Given any aggregator $A$, nodes in $\mathcal{P_A}$ reach the aggregator $A$ in one cycle.
\end{theorem}

\proof
Let $N$ be a node in $\mathcal{P_A}$. From Lemma~\ref{theo:dom}, this node has necessarily one dominator which is its next hop towards the aggregator $A$. Three cases are possible:
	\begin{enumerate}
	\item $N \in \mathcal{T}_A$: the result comes from Theorem~\ref{theo:treeNode}. 
	\item $N \notin \mathcal{T}_A$ and does not belong to any tree. According to Method~\ref{meth:orderColor}, the color of $N$ is ordered after nodes in the tree, and hence after its dominator. This dominator is able to reach the sink in a single cycle (see Theorem~\ref{theo:treeNode}). 
	\item $N \notin \mathcal{T}_A$ but belongs to a tree $\mathcal{T}_{A_1}$ rooted at another aggregator $A_1$. According to Lemma~\ref{lemma:otherTree}, $N$ is scheduled after its dominator, itself able to reach the sink in a single cycle according to Theorem~\ref{theo:treeNode}.
	\end{enumerate}	
\endproof

\section{$Highways$ Construction and Color Ordering}\label{sec:highways}
$Highways$ are paths between aggregators and the data sink. Data are routed from one aggregator to another in the direction of the sink. These paths are called $Highways$ because they are fast; they comprise a small number of nodes and slots granted to these nodes are ordered according to the appearance order of nodes in these paths. In this section, we detail these aspects.

\subsection{$Highways$ Construction and Routing over $Highways$}\label{sec:highCon}

The construction of $Highways$ is based on the following steps:

\begin{enumerate}
\item Considering the grid composed of aggregators (that is a grid where vertices are the aggregators), except aggregators on the grid border, other aggregators have $4$ neighboring aggregators. That is, if we consider the aggregator $A_0$ placed at the center of the grid, this aggregator has $4$ neighboring aggregators $A_u$, $A_v$, $A_{-u}, A_{-v}$ placed respectively at nodes of coordinates $u$, $v$, $-u$ and $-v$ where $u$ and $v$ are the VCM vectors. 
\item We build $4$ shortest paths between the aggregator $A_0$ and each of its neighboring aggregators $A_u$, $A_v$, $A_{-u}, A_{-v}$. For this reason, we speak about $4$ $Highways$: $H_u$ on the direction of $u$ vector (between $A_0$ and $A_u$), the second one denoted $H_v$ in the direction of $v$ vector, the third one denoted $H_{-u}$ is in the direction of $-u$ vector and the fourth one denoted $H_{-v}$ in the direction of $-v$ vector.
\item Two different $Highways$ should not share the same nodes in order to balance the routing load among different nodes.

Figure~\ref{fig:highway-R3} illustrates the $Highways$ built by ORCHID for the node $(0,0)$ for $R=3$.

\begin{figure}[h!]
\centering
\includegraphics[width=2.7in]{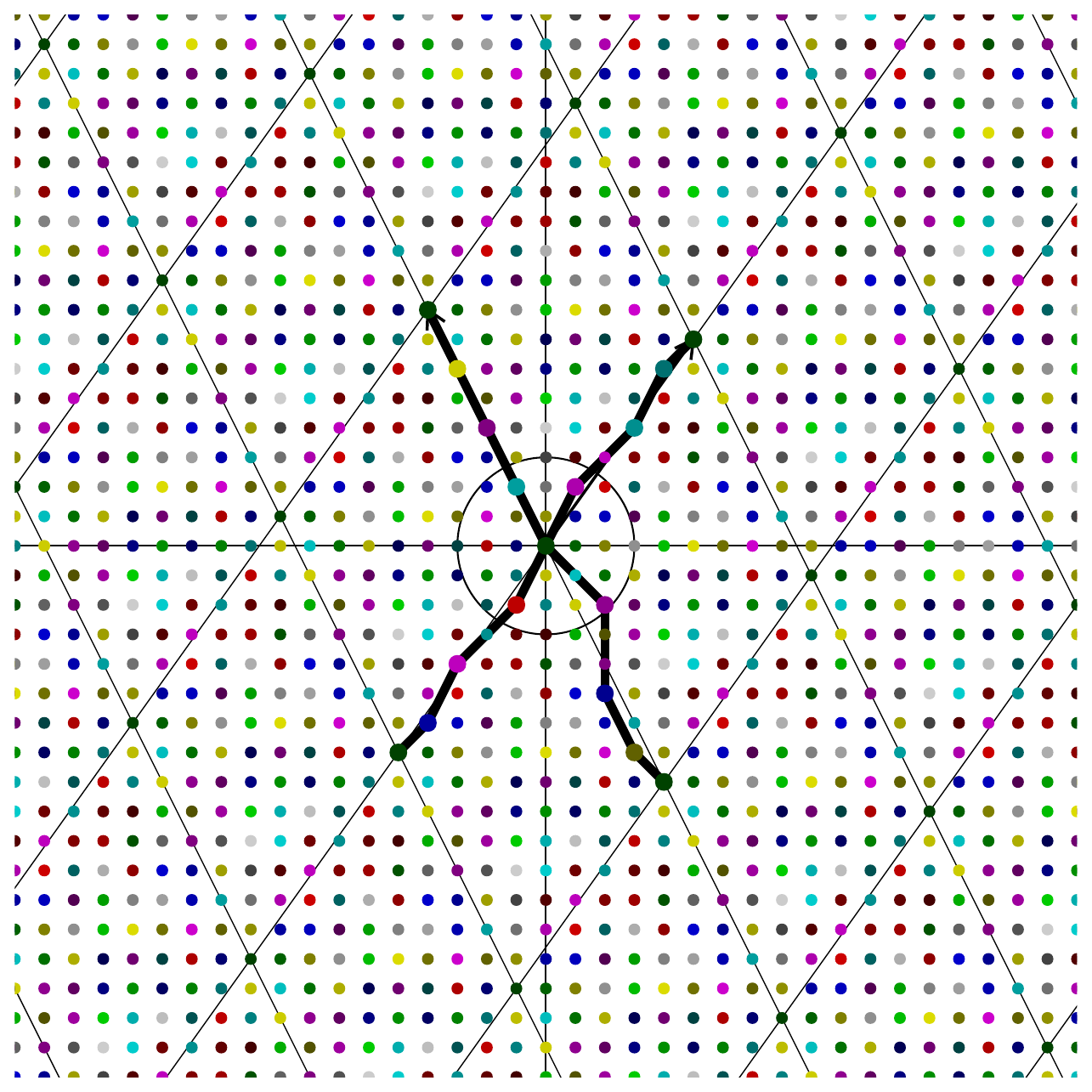}
\caption{Highways built by ORCHID for $R=3$.
\label{fig:highway-R3}}
\end{figure}

\item The $Highways$ built for the aggregator $A_0$ are repeated for all other aggregators. Notice that aggregators that are in the border of the grid do not need to use all of these $Highways$.
\item The previous steps allow each aggregator to have $Highways$ in four directions. These $Highways$ are used for routing as follows: We use a geographic routing like~\cite{geoRou} to select the closest neighboring aggregator to the sink. The path followed to reach this aggregator is the $Highway$ computed in step $2$.
\end{enumerate}

\begin{remark}
Each aggregator has $Highways$ in $4$ directions. This ensures a routing flexibility especially in the following scenarios: the position of the sink is unknown beforehand, the sink is mobile or when there are multiple sinks. Furthermore, it increases robustness against routing failures.
\end{remark}

\begin{remark}
Notice that building $Highways$ for further directions means activating more nodes in $H$ period. 
Hence, the number of built $Highways$ for any aggregator is subject to a compromise between delays and energy consumed per node. 
\end{remark}

\subsection{CO-Highways: Color Ordering for $Highways$}
After building $Highways$, we now order the colors of nodes in these $Highways$ in a specific part of the STDMA cycle dedicated to $Highways$. Indeed, for each $Highway$, we order the color of intermediate nodes of this $Highway$ such that the color of the $i^{th}$ node on this path is associated with the time slot number $i$ in the part of the STDMA cycle dedicated to $Highways$. 

\begin{theorem}
With $Highways$, no more than one cycle is required to any aggregator to reach one of its neighboring aggregators.
\end{theorem}
\proof
Nodes on $Highways$ have colors that are sorted according to the order of nodes in $Highways$. So, each node is scheduled before its next hop. Hence the theorem.
\endproof

\section{Orchestration in a STDMA cycle}\label{sec:orchestration}
The key idea of the proposed architecture is that each packet is transmitted first from its source node to the closest aggregator, then this packet is transmitted from one aggregator to another via $Highways$ until it reaches the data sink. 
In this section, we describe how to carry out this orchestration. 
\subsection{Overview}
Figure~\ref{fig:globalTDMA} depicts the global STDMA cycle obtained. This cycle is composed of two periods:
\begin{figure}[H]
\centering
\includegraphics[width=1.85in]{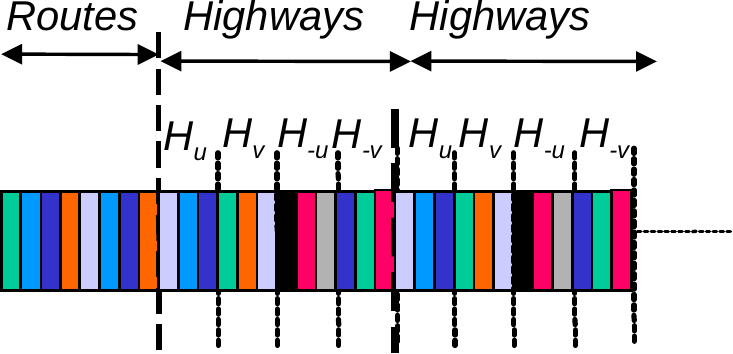}
\caption{Global STDMA cycle.\label{fig:globalTDMA}}
\end{figure}

\begin{enumerate}
\item The first period $R$ corresponds to the schedule 
of $Routes$. In this period, each node is awake during its slots and during the slots of its children in the dominating tree. Let $|R|$ be the number of slots of this period.
\item The second period $H$ is composed of colors of the $Highways$. This period is divided itself into $4$ sub-periods: $H_u$, $H_v$, $H_{-u}$ and $H_{-v}$. Each period $H_w$ corresponds to a $Highway$ between two aggregators optimized according to the direction of vector $w$ (see Section~\ref{sec:highCon}). The alternation of the $4$ directions in $H$ cycle favors access fairness between aggregators that follow different directions to reach the sink. 
The period $H$ is repeated a number of times until the farthest aggregator from the sink reaches this sink. In this period, each node belonging to any $Highway$ is awake during its slots and during the slots of nodes for which it is a next hop. 
\end{enumerate}

\subsection{Further CO-Highways Optimizations}\label{sec:highway-opt}
In the cycle for $Highways$, $H_u$, $H_v$, $H_{-u}$, $H_{-v}$, we observe
that one slot is allocated to aggregators at the beginning of each $H_k$ cycle where $k \in\{u,v,-u,-v\}$
. Some of these are unneeded, such as the first slot of $H_u$ 
that immediately follows the last slot of the $Routes$, also dedicated for
aggregators.
By further reordering the $Highway$ sequence as $H_u$, $H_{-u}$, $H_v$, $H_{-v}$,
the first slots of $H_{-u}$ and $H_{-v}$ can also be 
removed, as they never provide progress for routing\footnote{The one of $H_v$ is still required so as to be able to route both along $u$ and $v$ in the same highway cycle.}. Hence, if the highway $H_u$ has 3 intermediate nodes, the relative cycle should contain only 3 slots.

\section{Analytical Evaluation of ORCHID Performance} \label{sec:results}
The objective of this section is to estimate the number of slots of the global STDMA cycle of ORCHID and compare the  energy consumed in ORCHID and in IRCO 
(see Section~\ref{sec:IRCO}) when $R \rightarrow \infty$.
First, we introduce the following notations: \\
$\bullet$ The radio range is $R$ and the network is comprised of nodes of the grid, within the disk of size defined by $L$ (with radius = $L R$). \\
$\bullet$ The periodic coloring is defined with VCM vectors $u_1$, $u_2$; their angle
  is $\alpha$, and the number of colors is asymptotically equal to $\theta R^2(1+O(\frac{1}{R}))$ (Property~\ref{theo:nbColor}). \\
$\bullet$ We have proved in~\cite{VCM} that an asymptotically optimal (in number of colors) h-hop coloring produced by VCM yields 
  when  $R \rightarrow \infty$ values asymptotically equal to: 
  $\theta = \frac{\sqrt{3}}{2}h^2$, 
  $\alpha=\frac{\pi}{3}$, $|u_1| = |u_2| = hR$
 
The determination of the number of nodes in a disk of radius $R$ is known as the \textit{Gauss Circle Problem}\footnote{http://en.wikipedia.org/wiki/Gauss\_circle\_problem}. This number is equal to $\pi R^2
+o(R)$. In the following, we asymptotically approximate this number by $\pi R^2$.
\subsection{Cycle Length}\label{sec:analysis}
The STDMA cycle of ORCHID comprises a number of slots estimated by Theorem~\ref{theo:totalSlots}. 
\begin{theorem}\label{theo:totalSlots}
With VCM coloring, the global STDMA cycle comprises a number of slots estimated by $\theta R^2 + 4(h+1)\frac{4\sqrt{3}L}{3h}$~\footnote{without the optimization of Section~\ref{sec:highway-opt}.}.
\end{theorem}

\begin{proof}
The global STDMA cycle of ORCHID is composed of $Routes$ that asymptotically contains $\theta R²$ slots and $Highway$ that should be repeated a number of times that we denote $\lambda$. 
$\lambda$ should be such that one aggregator is able to reach the farthest other aggregator in the disk of radius $LR$. The largest \emph{geographic} distance between any two nodes corresponds to the diagonally opposite nodes with $2 L$. The minimum (worst-case) progress in along one direction with geographic routing is $\min(|u_1|,|u_2|) \sin(\alpha)$ per full $Highway$ cycle
 ($H_u$, $H_v$, $H_{-u}$, $H_{-v}$). Consequently, (when $R \rightarrow \infty$):
\begin{equation}
\lambda = \frac{2LR}{min(|u_1|,|u_2|) \sin(\alpha)} =\mathrm{(for~VCM)} \frac{4 \sqrt{3} L}{3h}
\end{equation} 
Now, we need to compute the number of slots of a given $Highway$. Remember that a $Highway$ is a path between two given neighbor aggregators. 
Asymptotically when $R \rightarrow \infty$ (and thus $|u_1|$, $|u_2|$ also 
grow), the number of hops approaches the distance, and is $\lfloor \frac{max(|u_1|,|u_2|) }{R}\rfloor + 1$.
For VCM, asymptotically, they are exactly $h+1$ hops away (see \cite{VCM}). 
This completes the proof.
\end{proof}

\begin{corollary}\label{coro:cycles}
Asymptotically, ORCHID with VCM coloring needs $1$ cycle of length $\theta R^2$ and $4 \times \frac{4 \sqrt{3} L}{3h}$ cycles of length $(h+1)$.
\end{corollary}

Notice here that the $Highway$ cycle is shorter than the $Routes$ cycles. We will find the same result with simulation results in Section~\ref{sec:slotOrchSimul}.

\subsection{Energy Consumption}\label{sec:energyCon}
The aim of this section is to evaluate the benefits of ORCHID from energy point of view. 
\newcommand{\cE}{\mathcal{E}}
Notice that generally with radio duty cycling scheduling, the sensor node consumes energy: during its slot if it is transmitting and during the slots of its potential transmitters. For simplicity, we will assume that such active node (during a slot), receiving or transmitting, will require an energy equal to $\cE$.
From results in Section~\ref{sec:analysis} and in Section~\ref{sec:model}, we can estimate the total energy consumed by all nodes in ORCHID and IRCO. \\

In Property~\ref{prop:eqModel}, we have evaluated the average delay in number of slots needed to traverse one radio range which is equal to $\frac{3}{2}\theta + \frac{3}{4}\pi$. In a disk of radius $LR$, the farthest two nodes are of a distance $2LR$. Hence, to route between any two nodes
in this disk, the total number of slots needed is: $ 2L  (\frac{3}{2}\theta + \frac{3}{4}\pi)$. Thus, a number of cycles equal to $ 2L  (\frac{3}{2}\theta + \frac{3}{4}\pi) \times \frac{1}{\theta R^2}$.

We assume that the density defined as the number of nodes per surface unit is equal to $1$. Hence, the total number of nodes in the disk of radius $LR$ is equal to: $\pi L^2R^2$. Each node is awake one slot per cycle for transmission and its neighbors ($\pi R^2 - 1$ neighbors) are awake during this slot for reception. Hence, the total energy consumed in IRCO is:
\begin{equation}
E_\mathrm{IRCO} = 2L\times(\frac{3}{2}\theta+\frac{3}{4}\pi)\times (\pi R^2)\times (\pi R^2 L^2) \times \frac{1}{\theta R^2} \times \cE 
\end{equation}
Recall that $\theta=\frac{\sqrt{3}}{2}h^2$, hence we have:
\begin{equation}
E_\mathrm{IRCO} = \frac{\pi^2 \sqrt{3}(\pi + h^2 \sqrt{3}) R^2 L^3}{h^2}\cE
\end{equation}

For ORCHID, the total energy consumed is the energy consumed in the $Route$ cycle plus the energy consumed in the $Highway$ cycle.  \\
Asymptotically, the $Route$ cycle contains $\theta R^2$ slots. For each slot, all nodes having the associated color (whose number is = $\frac{1}{\theta R^2} \times \pi R^2L^2$) are awake as well as their neighbors ($\pi R^2 - 1$). \\ 
From Corollary~\ref{coro:cycles}, $Highway$ cycles are repeated $4 \times \frac{4 \sqrt{3} L}{3h}$, each of length $h+1$ slots. During each $Highway$ slot, $2$ nodes are active (the transmitter and its next hop). Of course, all nodes having the same color as these two nodes (with a number = $ \frac{1}{\theta R^2} \times \pi R^2 L^2$) are also active, hence:
\begin{equation}
E_\mathrm{ORCHID} = \left( \frac{1}{\theta R^2} \times \pi R^2 L^2 \times \theta R^2 \times \pi R^2 + 4 \times \frac{4 \sqrt{3} L}{3h} \times (h+1) \times 2  \times    \frac{1}{\theta R^2} \times \pi L^2 R^2 \right) \cE
\end{equation}
Hence,
\begin{equation}
E_\mathrm{ORCHID} = (\pi^2R^4L^2 + \frac{32 \pi \sqrt{3} (h+1)}{3 h \theta}L^3)\cE
\end{equation}

As a result, when $R$ is fixed and when $L \rightarrow \infty$: most of
the cost from $E_\mathrm{ORCHID}$ comes from the $Highways$ (right part),
and the ratio
between the  two methods converges to 
$E_\mathrm{ORCHID}/E_\mathrm{IRCO} \rightarrow \frac{a}{R^2}$ where $a$ is a constant, illustrating the gains of ORCHID for higher densities.

\section{Experimental Results}\label{sec:exp-res}

We developed Python programs implementing ORCHID and IRCO. 
We compare these two methods in terms of delays and energy consumption.
\newcommand{\OO}{$\mathcal{O}$}
We used a topology with nodes of the grid inside the disk of radius $L=300$. 
The scenario is that information sent by any node in this topology reaches the aggregator \OO, the center of the disk.
\footnote{They could be extended to the case where every node reaches one aggregator.}.

\subsection{Number of Slots of ORCHID and IRCO}\label{sec:slotOrchSimul}
\subsubsection{How the Number of Slots is Computed?} 
\noindent $\bullet $ {\bf ORCHID}: 
We first find $\lambda$, the number of $Highways$ repetition. Then, the number of slots is obtain by adding this to $\lambda *$ size of $highways$ cycle  $+$ the size of the $Route$ cycle (= the number of VCM colors). \\
$\bullet $ {\bf IRCO with Shortest-Delay Path}: Path delays are expressed as in Equation~\ref{eq:delay}.

\subsubsection{Results}
We run experiments with various radio ranges ($1$ experiment par radio range). Table~\ref{tab:orchid-result} gives the worst case results (in terms of number of slots). Hence, this table gives an idea about the dimensioning of the global STDMA cycle. 

\begin{table}[!h]
\centering
\begin{scriptsize}
\begin{tabular}{|c|c|c|c|c|c|c|}\hline
\multicolumn{4}{|c|}{VCM parameters ($Route$)} & IRCO &
    \multicolumn{2}{c|}{ORCHID} \\ \hline
Radio range & $u_1$ & $u_2$ & nb colors. & max nb VCM cycles  & nb highway slots &  max nb highway cycles\\ \hline
$2$& $(4, 3)$ & $(-3, 4)$& $25$ & $62$ & $13$ & $60$\\ \hline
$3$& $(5, 7)$ & $(-4, 8)$& $68$ & $16$ & $13$ & $39$\\ \hline
$4$& $(8, 8)$ & $(-3, 11)$& $112$ & $9$ & $13$ & $30$\\ \hline
$5$& $(15, 3)$ & $(4, 14)$& $198$ & $5$ & $13$ & $23$\\ \hline

\end{tabular} 
\end{scriptsize}
\caption{Comparison of results for sample grid networks with various
radio range\label{tab:orchid-result}}
\end{table}
%

From Table~\ref{tab:orchid-result} we can deduce two results. \textit{First}, as highlighted in Section~\ref{sec:energyCon}, in ORCHID, the number of slots produced by the repetition of the $Highways$ is smaller than the number of slots in the $Route$ cycle.
For instance, for radio range $=4$, $Route$ cycle comprises $112$ slots and the $Highway$ cycle should be repeated $30$ times and comprises $13$ slots (with optimisation of Section~\ref{sec:highway-opt}, otherwise $16$): hence a total of $390$ slots. 
\textit{Second}, ORCHID has a smaller number of 
slots than IRCO. For instance, for radio range =$4$, while ORCHID uses $390+112=502$ slots, IRCO uses $112 \times 9=1008$ slots.


\begin{figure}[!h]
\centering
\includegraphics[width=0.7\linewidth]{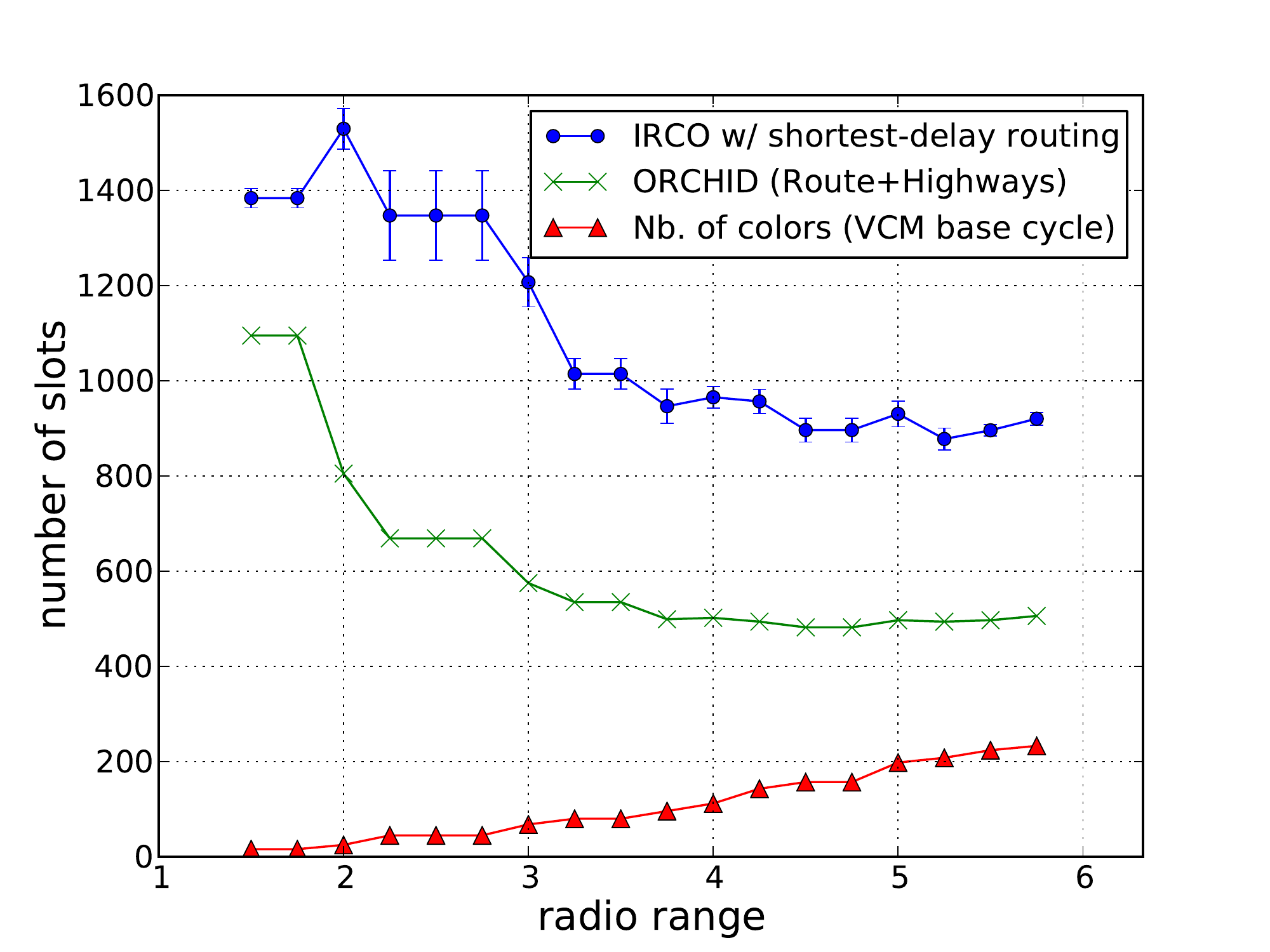}%
\caption{Total number of slots necessary for reaching center node in one cycle}\label{fig:orchid-result}
\end{figure}

Figure~\ref{fig:orchid-result} represents graphically the same information (average of $60$ simulations for IRCO): it is the total
number of slots necessary to guarantee that every node of the disk of radius $LR$ reaches the center aggregator \OO.

As illustrated in Figure~\ref{fig:orchid-result}, ORCHID is noticeably more efficient than IRCO. This highlights the efficiency of ORCHID scheduling. 

Notice also that as the radio range increases, the number of colors in the VCM cycle becomes the main part of the ORCHID cycle. For IRCO, this part would suffice to cover a large set of the nodes which means that this cycle is repeated only a few times in these topologies. Consequently, the margin of improvement is restricted in this case.

Note also that in actual WSN networks it is unlikely that a ``all-to-all'' route computation would be performed, hence another routing
such as IRCO with greedy routing would be implemented, at the
cost of worst performance than shortest-delay path. This means that ORCHID is suitable for WSNs.

\subsection{Energy Consumption in ORCHID}
\subsubsection{How Energy Consumption is Computed?}
The energy evaluation is computed as follows:
we count $1$ unit energy consumed 
for every slot where one node is listening to neighbors or is transmitting.
In a VCM cycle ($Route$ cycle), potentially every node should be listening
to every neighbor (as it might be used as a relay by these neighbors). In a $Highway$ cycle, in each slot, only nodes belonging to a $Highway$ can be transmitting and only their predefined destination should be listening.

\subsubsection{Results}
\begin{figure}[!h]
\centering
\includegraphics[width=0.7\linewidth]{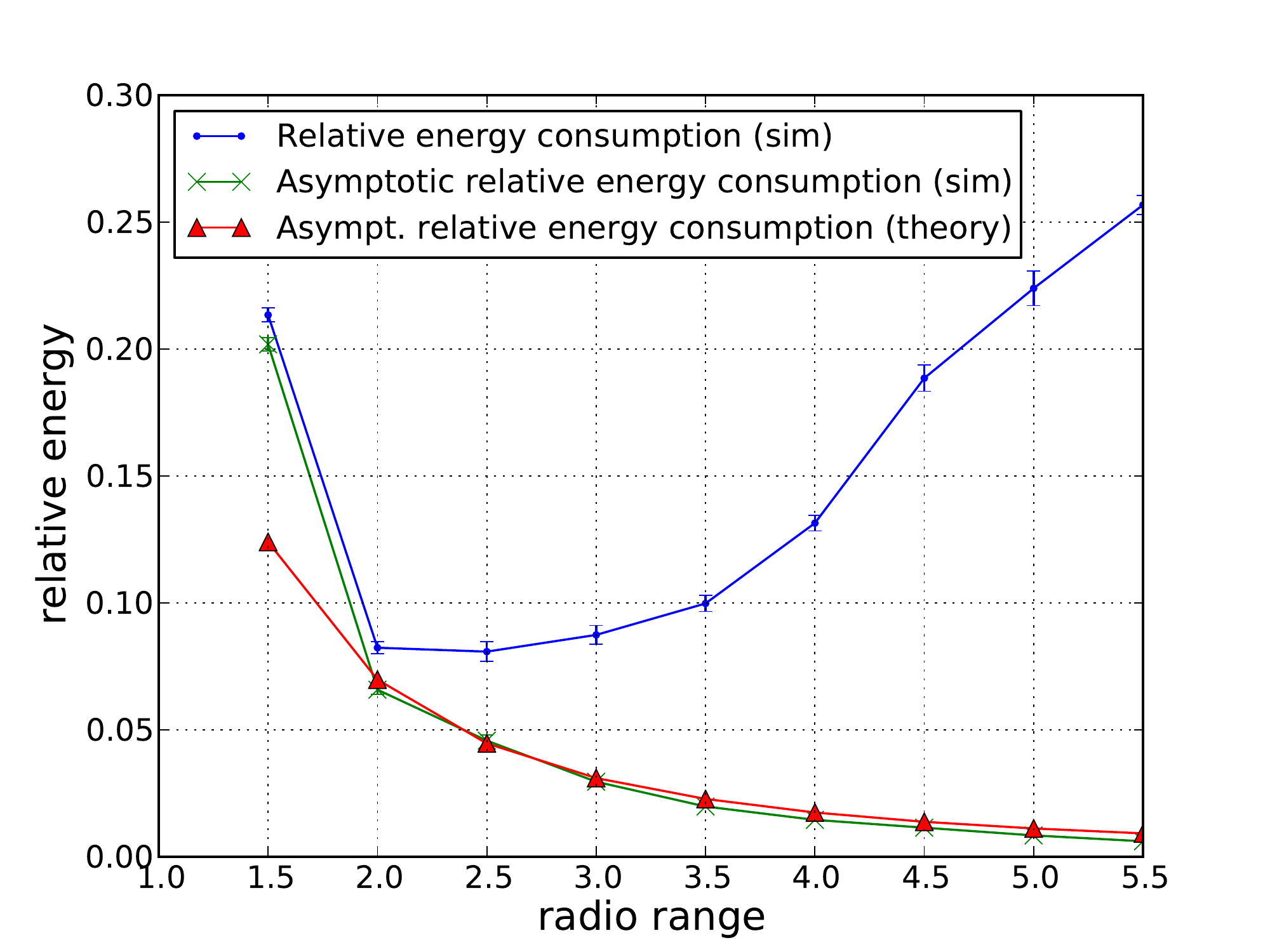}%
\caption{Ratio between energy needed by ORCHID, and by IRCO}\label{fig:energy}
\end{figure}
Although ORCHID requires less cycles for guaranteeing the same routing to one aggregator, another essential feature of ORCHID is the energy-efficiency. This is illustrated by Figure~\ref{fig:energy}: the energy consumption is evaluated for both methods, with schedule parameters  found as in Figure~\ref{fig:orchid-result}. 

In Figure~\ref{fig:energy}, we show the ratio between the energy 
consumed with ORCHID and the energy consumed by IRCO. It is represented by the curve ``Relative energy consumption'': it appears that this ratio is between 
$0.1$ and $1/3$: this means that ORCHID consumes $3$ to $10$ times less
energy than repeating VCM, a dramatic improvement.

Moreover, as detailed previously (Section~\ref{sec:energyCon}), the asymptotic energy gain of ORCHID is 
even better: this is because the density of the $Highway$ is much lower
than the density of the network (especially when the radio range increases).
Hence it is very beneficial to replace repetitions of the VCM cycles by
repetitions of the $Highway$ cycles. In Figure~\ref{fig:orchid-result},
we extrapolate the simulation results for estimating the ratio when
the radius of the disk grows to infinity. Under this condition, in ORCHID, $Highways$ must be repeated several times. Hence, the $Route$ cycle would be a vanishingly low part of the energy cost compared to $Highways$. For this reason, we compare the energy consumed because of the repetition of $Highways$ compared to the repetition of $\mathrm{E\_{IRCO}}$.
This is illustrated by the curve ``Asymptotic relative energy consumption'':
as shown, energy consumption in ORCHID is $2$ orders of magnitude lower
than VCM cycle repetition for range greater than $4$.



\subsection{Number of Reachable Aggregators in one Cycle per Node}\label{sec:reachAgg}
Figure~\ref{avgNbAgg} illustrates the average number of aggregators that any node can reach in a single cycle. 
This figure shows that this number is large. For instance, for a radio range $R=4.5$, 400 aggregators are reachable in one cycle. Consequently, a node can route through any of these aggregators in one cycle. This is useful in case of multiple sinks or mobile sink. In these scenarios, the closest aggregator to a given sink is not necessarily the best aggregator to reach another sink, or the same sink in a new location. Hence, having many possible aggregators that a node can reach rapidly (in one cycle)    is a real advantage of ORCHID.
This is also useful in case of link failure with a given aggregator. To conclude, ORCHID supports sink mobility and some link failures.

\begin{figure}[h!]
\centering
\includegraphics[width=3.5in]{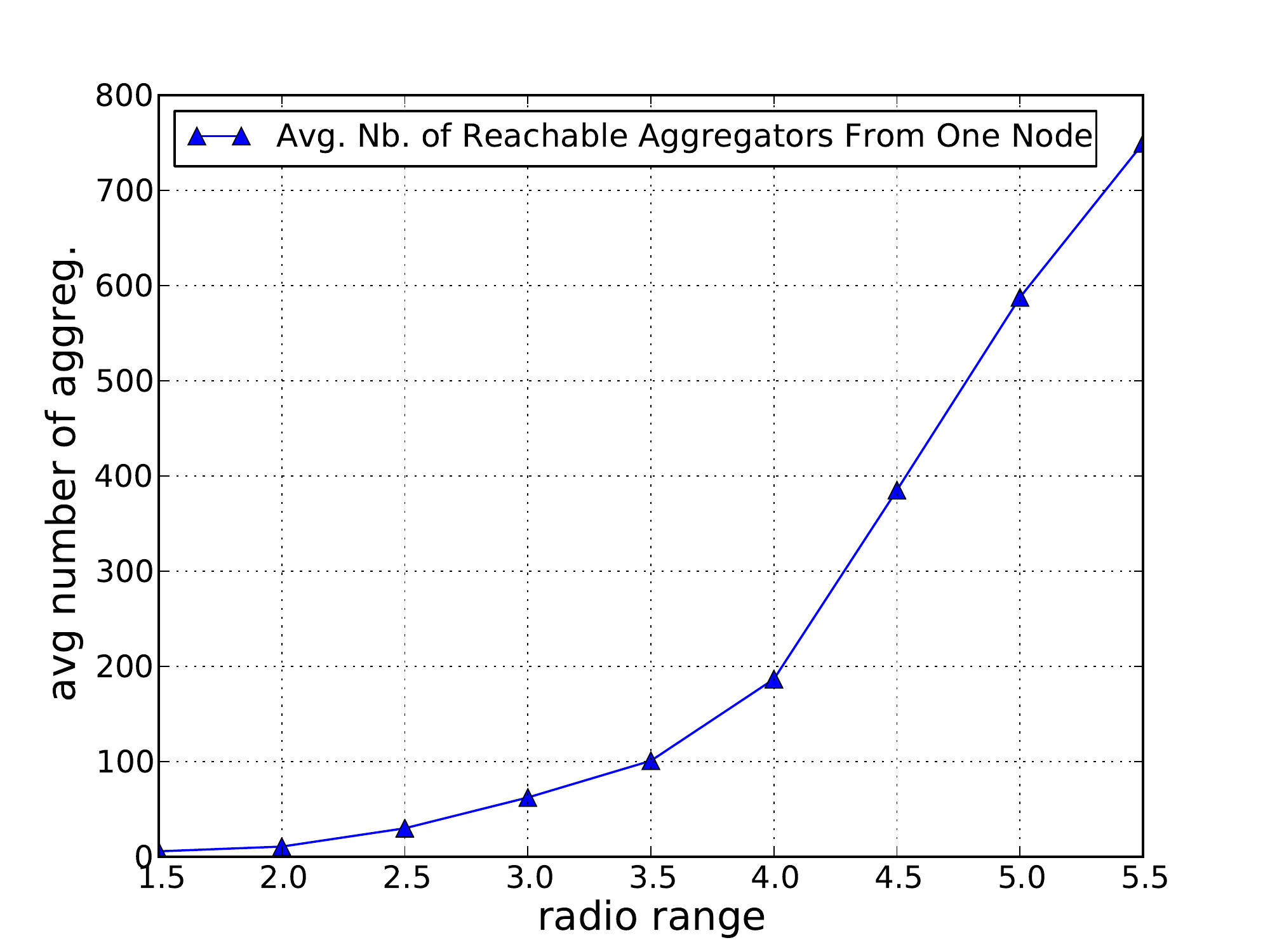}
\caption{Average number of reachable aggregators per node.
\label{avgNbAgg}}
\end{figure}




\clearpage
\part{Estimate of Normalized Delay per Range with Random Color Ordering}\label{part:model}

In this part, we design a model for the determination of the normalized delay per range when the radio range~$\rightarrow \infty$. 
We consider the same system model of Section~\ref{sec:genAssum} and in addition, we make the assumption that the colors are randomly associated with slots. 

\section{Model} \label{sec:modelProof}\label{sec:model}

We follow the approach commonly used in geographic routing studies,
assuming that a route from a source to a destination is followed.
Then the expected progress towards the destination is computed on one step
at an intermediate relay node (as in \cite{kleinrock} for instance). 
Greedy routing is assumed: but here, unlike geographic routing, the
selected next hop is not the one that is the closest to the destination, 
but the one that minimizes the ``normalized delay per range'' (as in \cite{CMAC}).

Let $N$ be this intermediate node. Without loss of generality, we assume that its coordinates are $(0,0)$ and its color is $0$. We denote $n$ the radio range.
Let $\mathcal{N}_{u}$ the set of neighbors of $N$ with x-coordinate $u$ where $u$ in an integer $\in [0,n]$. Let $y(u)$ be the maximum y-coordinate of these nodes. 
We have $y(u) = \lfloor\sqrt{n^2-u^2}\rfloor$. Also, asymptotically, we can write: 
\begin{equation}\label{eq:y}
y(u) = \sqrt{n^2-u^2} + O(1) 
\end{equation}
%
The colors of these nodes are denoted $c_{u,-y(u)}, \ldots, c_{u,1}, c_{u,2},\ldots, c_{u,y(u)}$. Figure~\ref{fig-model} illustrates an example. For clarity reasons, we only present the colors of nodes in the quarter of the disk. 

\begin{figure}[h!]
\centering
\includegraphics[width=2.7in]{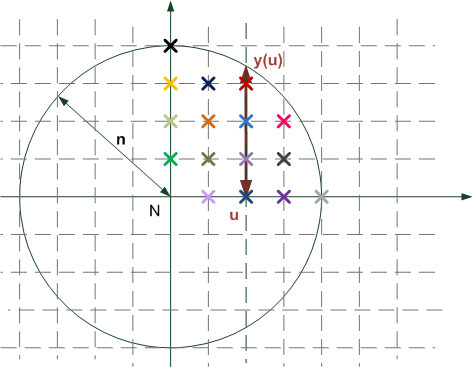}
\caption{Figure illustrating the model used.\label{fig-model}}
\end{figure}

We assume the following approximations: 1) that the destination lies on the x-axis
(towards infinity), 2) that the colors are i.i.d (independent and identically distributed) random real values 
(not integer) uniformly distributed in the interval  $[\alpha, \beta]$.
We select as next hop the neighbor with coordinates
$(u,i)$ that minimizes the ``normalized delay per range'', that is the quantity 
$\frac{c_{u.i}}{\frac{u}{n}}$.

Let $Z \triangleq \max\limits_{(u,i) \in N_u} \frac{c_{u.i} n}{u}$ be this quantity.
It is a random variable, and we show that: 
\begin{property}\label{prop:model}
For $n \rightarrow \infty$, the distribution of $Z$ is approximated by a shifted exponential distribution that has expected value of $\frac{\alpha\pi}{4}+\frac{3(\beta-\alpha)}{2n^2}$.
\end{property}




\section{Computation of the Normalized Delay Per Range}
Let now prove Property~\ref{prop:model}. We introduce intermediate random variables: \\
$\bullet$ \textit{$Y_{u,i} \triangleq \frac{c_{u,i} n}{u}$ is the
 normalized delay per range when selecting neighbor $(u,i)$ 
 with color $c_{u,i}$.}
 Because $c_{u,i}$ is uniformly distributed in $[\alpha,\beta]$, 
    its cumulative distribution function is 


\begin{eqnarray}\label{eq:Yui}
F_{Y_{u,i}}(x) = Pr( Y_{u,i}\le x) = \sigma(\frac{\frac{xu}{n}-\alpha}{\beta-\alpha}) \label{eq:FY}
\end{eqnarray}
with $\sigma$ ``bounding function'', defined as follows:   $\sigma(x) = x$ for $0 \le x \le 1$, $\sigma(x) = 0$ for $x<0$, and $\sigma(x) = 1$ for $x > 1$. \\
$\bullet$ \textit{$X_u \triangleq \min\limits_{i=-y(u), \ldots, +y(u)}~Y_{u,i} $}
is the normalized  delay per range when choosing about the best neighbor with x-coordinate $u$.
Its distribution verifies:
\begin{eqnarray}
F_{X_u}(x)  =   {Pr(X_u \le x)}  
 =  1 -  (1-Pr(Y_{u,0} \le x)) ^{(2y(u)+1)} \label{eq:FX}
\end{eqnarray} \\
$\bullet$ The random variable $Z$  can also be expressed as a minimum 
of minima: $Z= \min\limits_{u \in \{0, 1, \ldots n\}}~X_{u}$. Using Equations~(\ref{eq:FY}) and (\ref{eq:FX}), it is possible to express 
its cumulative distribution $F_Z(x)$ as:
\begin{eqnarray}
F_{Z}(x)  &=   {Pr(Z \le x)}  
 =  1 - \prod\limits_i
(1-Pr(X_{u} \le x)) \notag \\ 
\log(1-F_{Z}(x)) &=\sum_{u=1}^{n}(2y(u)+1)\log(1-\sigma(\frac{\frac{xu}{n}-\alpha}{\beta-\alpha}))
\end{eqnarray}

\newcommand{\tZ}{\tilde{Z}}

If we approximate $\log(1-a) \approx -a$ for $a$ near $0$ (also an upper bound in all cases), and
if we ignore bounding effects of $\sigma$, we can approximate $Z$ 
by $\tZ$, with a close distribution function (when $n \rightarrow \infty$):%
\begin{eqnarray}
\log(1-F_{\tZ}(x))
= -\sum\limits_{u=1}^{n}(2y(u)+1)\frac{\frac{xu}{n}-\alpha}{\beta-\alpha} \notag
\end{eqnarray}
Using Equation~\ref{eq:y}, we have:
\begin{eqnarray}
\log(1-F_{\tZ}(x)) = -\frac{2}{\beta-\alpha} \left( \frac{x}{n}\sum_{u=1}^{n}{u\sqrt{n^2-u^2}} - \alpha \sum_{u=1}^{n}{\sqrt{n^2-u^2}} +O(n) \right) \notag
\end{eqnarray}

Now we introduce the following approximation:

\begin{equation*}
\sum_{u=0}^{n}{\sqrt{n^2-u^2}} = \frac{\pi n^2}{4} + O(n)
\end{equation*}
\begin{equation*}
\sum_{u=0}^{n}{u\sqrt{n^2-u^2}} = \frac{n^3}{3} + O(n^2)
\end{equation*}
Indeed, these are sum of (bounded) functions, respectively: non-decreasing 
and with one unique maximum (increasing then decreasing).
It is easy to bound them with integrals:
\begin{eqnarray*}
 (\int_0^n \sqrt{n^2-u^2}~\mathrm{d}u) -n  & \le & \sum_{u=0}^{n}{\sqrt{n^2-u^2}} \le (\int_0^n \sqrt{n^2-u^2}~\mathrm{d}u)+n   ~~(\mathrm{since}~\sqrt{n^2-u^2} \le n) \\
( \int_0^n u\sqrt{n^2-u^2}~\mathrm{d}u) - 3 n^2 & \le & \sum_{u=0}^{n}{u\sqrt{n^2-u^2}} \le (\int_0^n u\sqrt{n^2-u^2}~\mathrm{d}u) + 3 n^2  ~~(\mathrm{since}~u\sqrt{n^2-u^2} \le n^2)
\end{eqnarray*}
Hence,

\begin{eqnarray}
\log(1-F_{\tZ}(x)) = -\frac{2n^2}{\beta-\alpha} \left(  \frac{x}{3} - \alpha \frac{\pi}{4} + O(\frac{1}{n})\right)
  \mathrm{(bounding~sums~by~integrals)} \notag
\end{eqnarray}

\newcommand{\ttZ}{\widehat{Z}}
Further ignoring the $O(\frac{1}{n})$, we approximate $\tZ$ by $\ttZ$ with:
\begin{eqnarray}\label{eq:PZFinal}
Pr(\ttZ \le  x ) = 1- exp\left[ -\frac{2n^2}{3(\beta-\alpha)} (x - \frac{3\alpha\pi}{4})\right]
\end{eqnarray}
Equation~\ref{eq:PZFinal} meaningful for $x\ge\frac{3\alpha\pi}{4}$.
Equation \ref{eq:PZFinal}, corresponds to a variable $\ttZ-\frac{3\alpha\pi}{4}$ exponentially distributed with parameter $\lambda=\frac{2n^2}{3(\beta-\alpha)}$.
Therefore, the expectation of $\ttZ$, $E(\ttZ)$ is:%
\begin{equation}
E(\ttZ)=\frac{3\alpha\pi}{4}+\frac{1}{\lambda}=\frac{\alpha\pi}{4}+\frac{3(\beta-\alpha)}{2n^2}  \label{eq:EZ}
\end{equation}
$\ttZ$ is an approximation of $Z$ and hence  Property~\ref{prop:model}.

\section{Experimental Validation of Asymptotic Expression for Normalized Delay per Range}
In the previous section, an asymptotic expression was obtained for the
normalized delay per range. 
The goal of the present section is compare this asymptotic expression
with actual results obtained by simulations of the exact modeled system:
this experimentally validates the convergence to the asymptotic value, and also some of the approximations that were made.
\begin{figure}[!h]
\centering
		\subfigure[Cumulative distribution in the form:~$\log(1-F(x))$]{\includegraphics[width=3in]{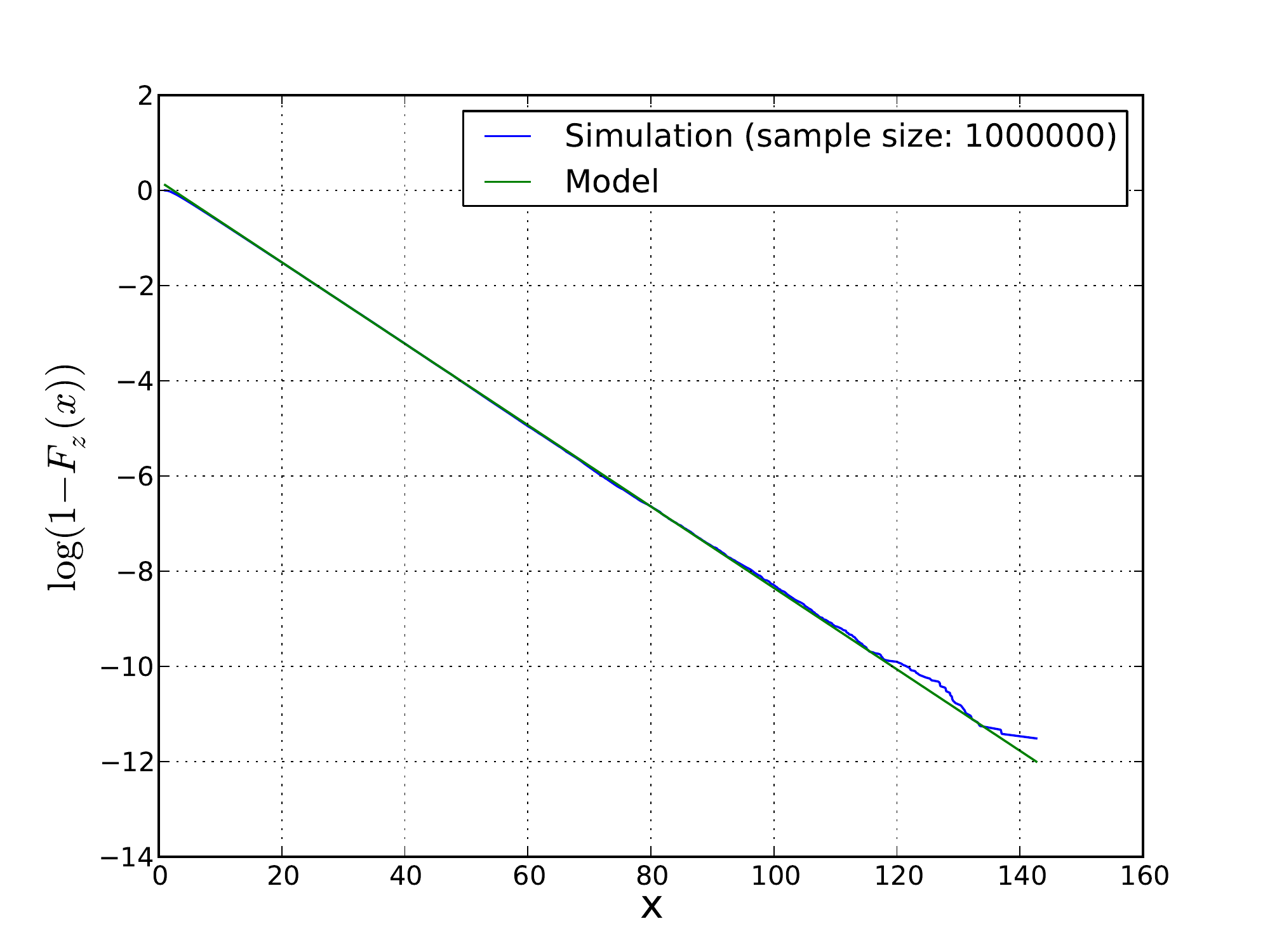}\label{fig:model-sim-distrib}}%
		\subfigure[Convergence to asymptotic expression (magnified: y-axis does not start from 0)]{\includegraphics[width=3in]{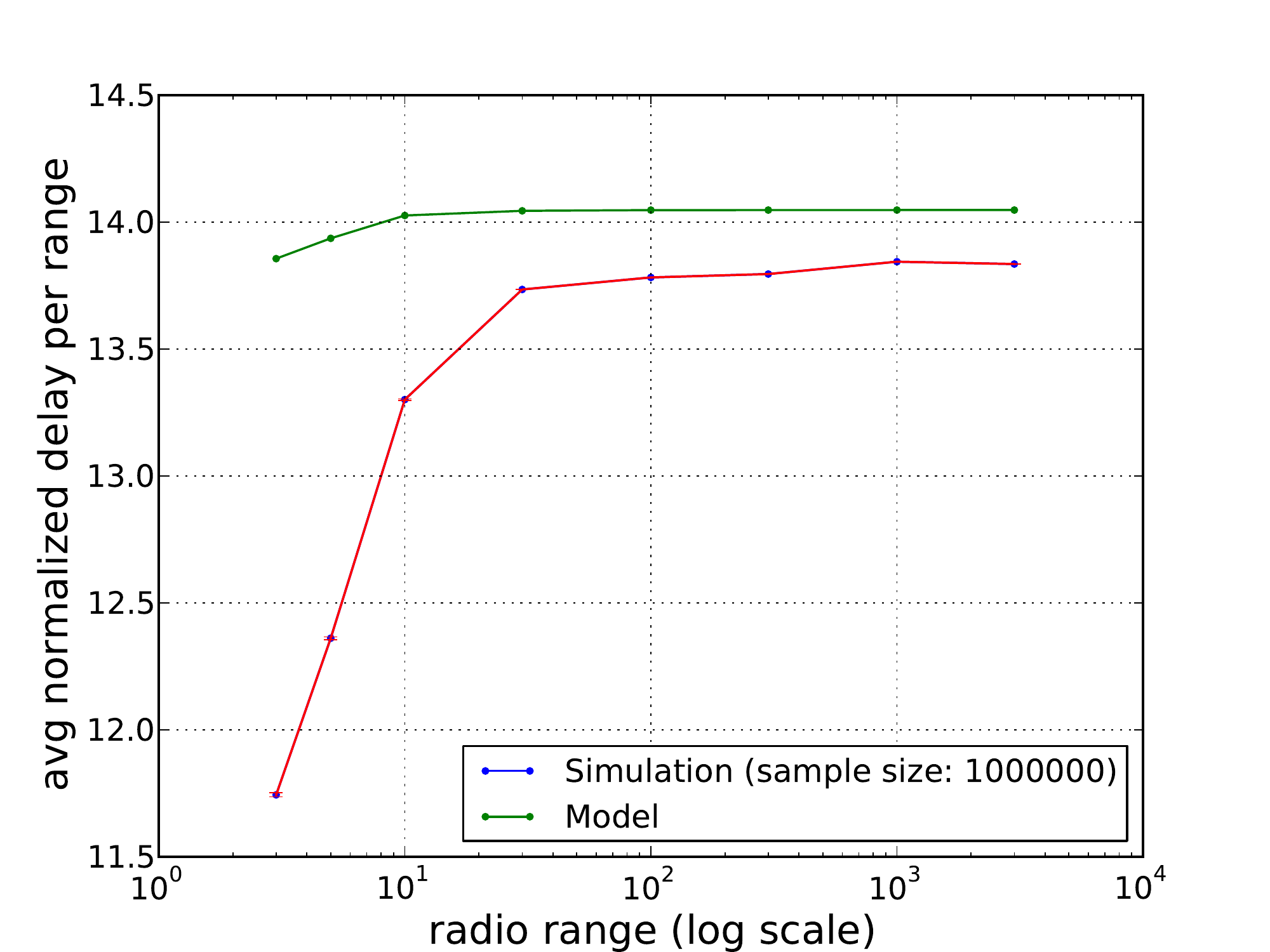}\label{fig:model-delay-vs-range}}%
		\caption{Experimental Validation of Model.\label{fig:model-validation}}
\end{figure}
\newpage
Figure~\ref{fig:model-sim-distrib} represents the results of 
computing the average value of normalized delay per range 
for a radio range of $1000$ 
(from one million random draws). The number of colors is equal to 
$\lfloor\theta n^2 \rfloor$ where $\theta = \frac{\sqrt{3}}{2}h^2$ ; 
this is asymptotic number of colors of VCM for a $3$-hop coloring, 
see~\cite{VCM}. As shown in the figure, parameters of the
exponential distribution match very closely the actual cumulative distribution.

The complementary Figure ~\ref{fig:model-delay-vs-range} represents 
the convergence of the simulation results to the asymptotic expression.
The chosen y-axis scale enhances the differences, but for a radio range equal 
to $3000$ we witness a difference of only $1.5 \%$ between model and simulations.



\begin{figure}[h!]
\centering
\includegraphics[width=3.3in]{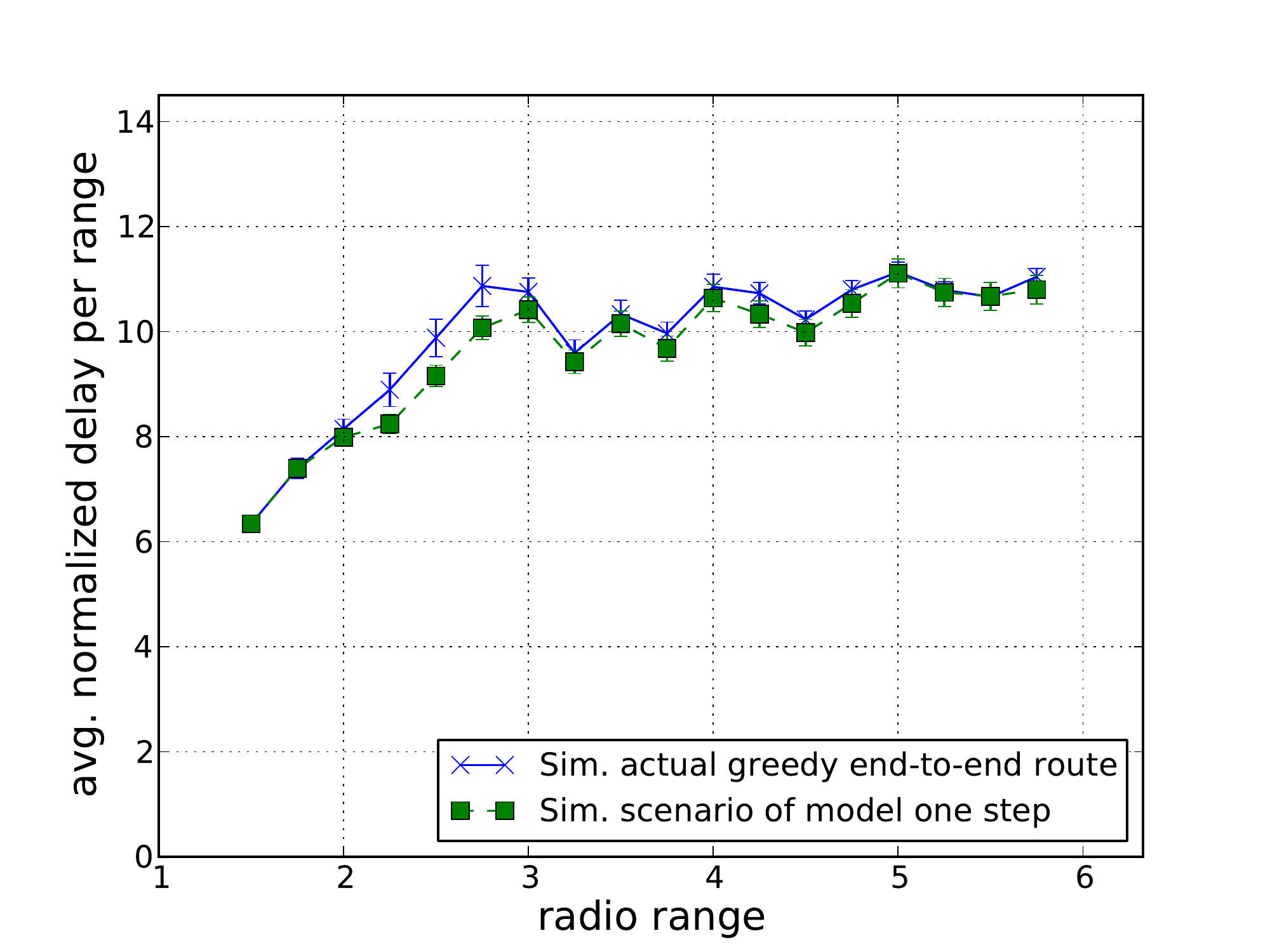}
\caption{Normalized delay in greedy routing and model.}\label{fig:delay-sim-model}
\end{figure}

Figure~\ref{fig:delay-sim-model} illustrates the average normalized delay per range computed by exactly simulating the scenario of the model, and also the simulations of Section~\ref{sec:exp-res} with greedy routing. 

\begin{enumerate}
\item ``Sim. scenario of model one-step'' represents the 
simulation of the scenario described in the model for one step (see Section~\ref{sec:modelProof}). Colors are assimilated to real random numbers.
\item ``Sim. actual greedy end-to-end route'' represents the simulation of the greedy end-to-end route. Recall that colors are randomly ordered in the cycle and that this greedy routing uses the same heuristic as the model (minimize ratio delay per progress). We determine the end-to-end greedy path and the corresponding end-to-end delay. The average normalized delay per range is then obtained by dividing this end-to-end delay by the number of hops of the path according to the definition. 
\end{enumerate}
The observed normalized delay is similar for both methods. This result validates the method and scenario adopted by our model to compute the average normalized delay per range. 

\newpage
\part{ORCHID for General Graphs under SINR Interference Model}\label{sec:orchid++}
So far, we have considered a grid network with a protocol-based interference model. In this part, we will discuss how to use ORCHID for general graphs under the SINR interference model (realistic physical model). 

\section{From General Graphs to Grids}\label{sec:realORCHID}
We assume that we have a general graph of a given density. We assume also that this graph is fully connected. Our aim is to profit from ORCHID benefits, in particular from VCM. Indeed, as previously highlighted, VCM needs no message exchange between nodes to perform grid coloring. Hence, to be able to apply grid coloring, our methodology is to map a cell grid over the general graph. This grid must contain all nodes. Its grid step (the cell width) must be $\le$ the transmission range of nodes. It is normalized to $1$. 
\begin{remark}
Notice that each cell may contain more than or equal to zero nodes. However, we assume that at each cell there is at least one node. 
In a future work, we will consider the case of empty cells and compensate the arising routing holes by local repair. Section~\ref{sec:reachAgg} has illustrated the following property: the number of aggregators reachable by one node in one cycle is large, and increases with radio range (or equivalently with density, through rescaling). Thus one simple local repair strategy would be to select a different route with a different aggregator.
\end{remark}  
Given this cell grid, we use ORCHID as follows: \\
$\bullet$ Instead of coloring nodes individually, our idea consists in coloring the cells using \textit{VCM++}: a modified version of VCM that is based on the SINR interference model (see section~\ref{sec:vcm-sinr}). Hence, each cell has a color computed by VCM++. Cells having the same color are cells at the lattice formed by the generator vectors of VCM++. \\ 
$\bullet$ Nodes inside one cell are assigned different colors. \\
$\bullet$ For each cell, we select a specific node called the \textit{cell representative}. It is the closest node to the center of the cell. Nodes inside a given cell transmit data to their cell representative. Hence, each cell has a slot called \textbf{super-slot} associated with its color. This slot itself is divided into \textbf{sub-slots} reserved to the colors of nodes inside the corresponding cell.
Figure~\ref{fig:vcm++} illustrates an example.
\begin{figure}[H]
\centering
\subfigure[Cell Grid]{\includegraphics[width=1.8in]{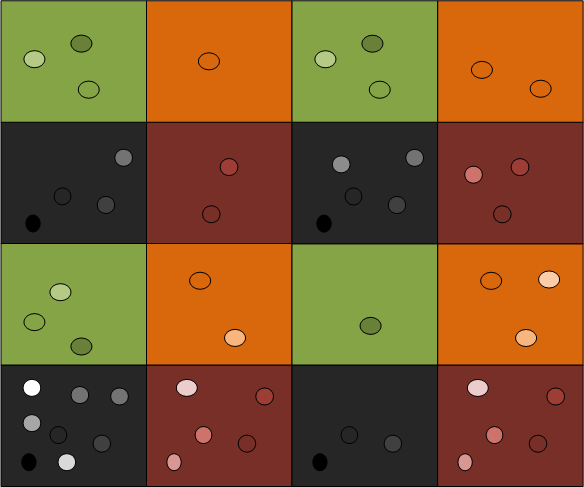}}
\subfigure[Routes Cycle]{\includegraphics[width=3.2in]{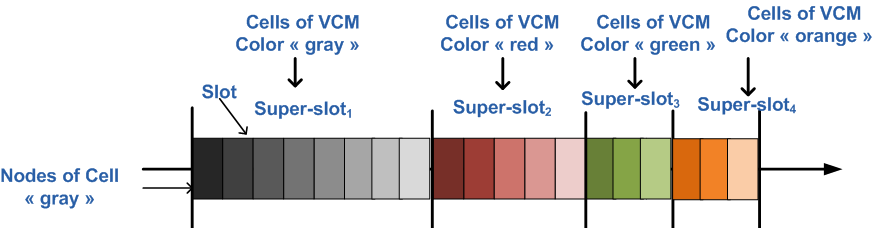}}\label{fig:cycle-new}
\caption{An illustrative example of a colored cell grid and the cycle obtained.\label{fig:vcm++}}
\end{figure}
In this figure, we have $4$ colors, hence $4$ super-slots. The most populated cell with color orange contains $3$ nodes, hence the super-slot orange is divided into $3$ sub-slots. \\

\section{Adapting ORCHID Routing}
Concerning routing, we build the hierarchical architecture of ORCHID as follows:
\begin{enumerate}
\item We consider the cells occupying the lattice $\mathcal{L}(u_1,u_2)$ where $u_1$ and $u_2$ are the generator vectors of VCM++. \textit{Aggregators} are representative of these cells.
\item \textit{Dominating Trees} are rooted at these aggregators. 
 Notice here that given two aggregators, dominating trees rooted at these aggregators are not necessarily similar because the initial general graph is no longer regular.
\item We use the same algorithm as in ORCHID to determine \textit{Highways} between any two given aggregators.
\end{enumerate}  

Furthermore, we keep \textit{CO-Ordering} of ORCHID to order the colors of nodes in $Routes$ and $Highways$. The global cycle is formed by $Routes$ period (see Figure~\ref{fig:vcm++} as an example) followed by $Highway$ period. 

In the following, we provide some details about the coloring method VCM++.

\section{VCM++: Adapting VCM to Color Cells using SINR Interference Model}\label{sec:vcm-sinr}
The physical interference model is defined as follows~\cite{gupta-kumar}.
\begin{definition} \label{def:sinr}
Let $P$ the common transmission power at all nodes. Let ${(X_k)}$, $k$ integer, be the set of nodes transmitting simultaneously. Let the constant $N$ be the background noise. A transmission from a node $X_i$ is successfully received by a node $X_j$ if:
\begin{equation}
\frac{\frac{P}{|X_i-X_j|^ \alpha}}{N+\sum_{k \ne i}{\frac{P}{|X_k-X_j|^ \alpha}}} \ge \beta
\end{equation}
where $|X_i-X_j|$ denotes the distance between any two nodes $X_i$ and $X_j$ and $\beta$ a constant threshold.
\end{definition}

Many works dealt with the scheduling in lattice based on SINR interference model (e.g. \cite{STDMA2,haenggi}). 
In our work, concurrent transmissions are at the lattice generated by VCM++ vectors. The aim of this section is to modify the method of vector search of VCM in case SINR interference model is considered. Let $\mathcal{L}(u_1,u_2)$ be the lattice generated by VCM++ generator vectors $u_1$ and $u_2$:
$\mathcal{L}(u_1,u_2) = \{(au_1+bu_2) | (a,b) \in \ZZ^2\}$.
For any vector $u$, we define the cell $c(u)$ as: $c(u) =\{ u+(x,y) | (x,y) \in [0,1]^2\}$. $u$ is called the coordinate of the cell. 
When we apply VCM++ to color the cells, all nodes in all cells of coordinates $w$ such that $w \in \mathcal{L}(u_1,u_2)$ will transmit at the same super-slot. 
Considering a transmission from node $(0,0)$ to node with coordinates  $z$, the interference at node of coordinate $z$ in the worst case is:
\begin{equation}
I(z) = \sum_{w\in \mathcal{L}(u_1,u_2)\setminus\{(0,0)\}}{\max_{t\in c(w)} P |t-z|^{-\alpha}}
\end{equation}
The associated SINR $S_{min}(w)$ computed at any cell $c(v)$ upon a reception from the cell $c(0)$ is:
\begin{equation}
S_{min}(v) = \min_{\substack{t_u\in c(0)\\ t_v\in c(v)}} \frac{P |t_u-t_v|^{-\alpha}}{I(v)+N}
\end{equation}

The transmission from $c(0)$ to $c(v)$ is successful if and only if: $S_{min}(v) \ge \beta $. 
 Then, to determine the successful transmissions from $c(0)$, we search the largest disk of radius $D_{u_1,u_2}$ where every cell can receive successfully from $c(0)$. \\
 Hence,  
$D_{u_1,u_2} = \mathrm{max} \{ |v|: v \in \ZZ ^2 \mathrm{~,~} \forall w \in \ZZ ^2:|w|>|v| \implies S_{min}(w) < \beta \}.$

Finally, we can select $u_1$ and $u_2$ as:
\begin{equation}
u_1,u_2 = \argmax_{(u_1,u_2)}{\frac{\pi D_{u_1,u_2}^2}{\mathrm{det(u_1,u_2)}}}
\end{equation}
With this formula, we ensure a tradeoff between the number of successful transmissions and the total number of colors.

\clearpage
\part{Conclusion}\label{part:conclusion}\label{sec:conclusion}
This report addressed the delay optimization issue in STDMA systems. In particular, we focused on delays induced by the misordering of time slots which is less often treated in literature. 
To deal with this problem, we first started by evaluating the delays when the slots are randomly ordered in the cycle. Hence, we evaluated the normalized delay per range  
with simulation and with a stochastic model.
Then, we proposed a solution called ORCHID that joins routing and scheduling. 

The challenge was to build routes and then to order nodes transmissions according to their appearance order in these routes. This ensures that data are delivered in a single STDMA cycle, which is very delay efficient. We performed analytical and experimental performance analysis of ORCHID. Results show that ORCHID uses smaller number of slots and is more energy efficient than shortest-delay path routing.

We also provided a method to adapt ORCHID to general graphs under the SINR model. The key idea was to map a cell grid over the initial graph. Considering the assumption that any cell can be empty is a future work.

%
%

\clearpage
\section*{Acknowledgment}
{This work has been partly supported by DGA/ASTRID/ANR-11-ASTR-0033.}

\clearpage
\part{Annex}
\begin{appendix}
In this annex, we provide figures illustrating some results. 

Figure~\ref{domTree-R2} illustrates a dominating tree built by ORCHID for $R=2$.
\begin{figure}[h!]
\centering
\includegraphics[width=2.7in]{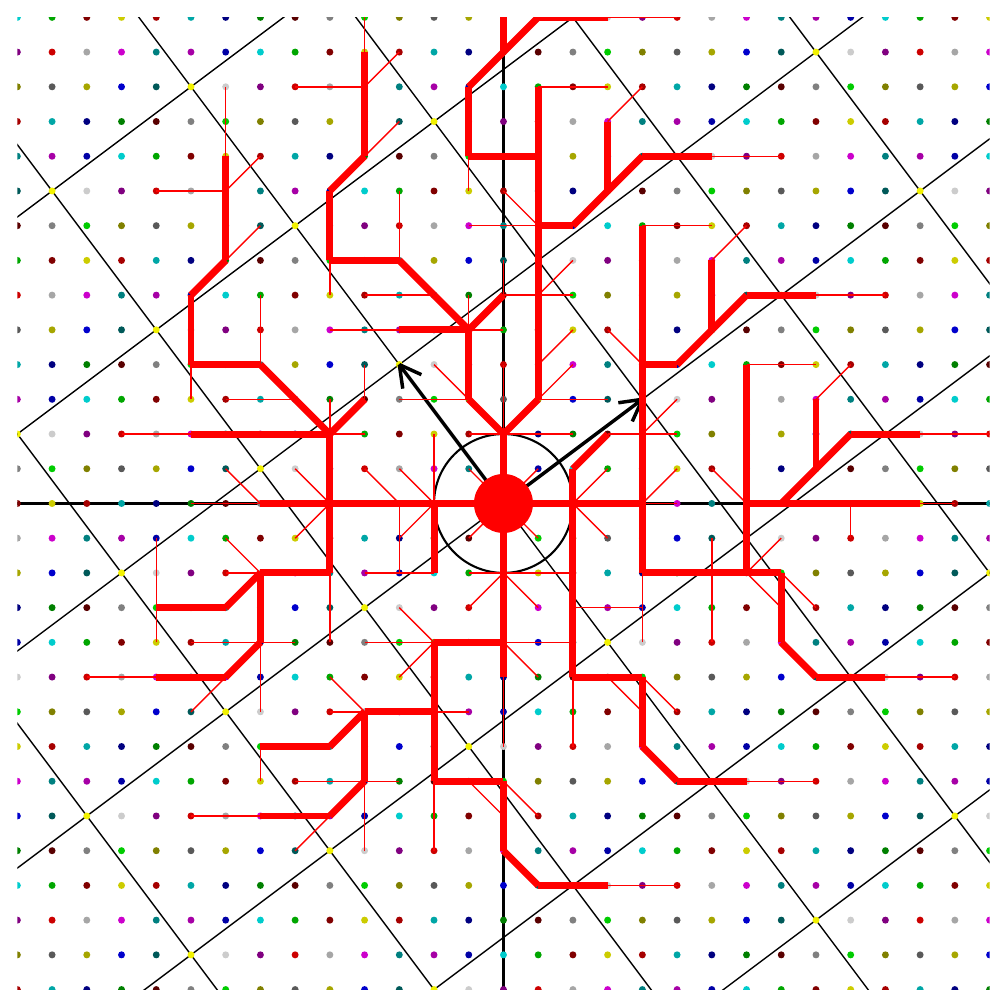}
\caption{The dominating tree generated by ORCHID 
for R=2.\label{domTree-R2}}
\end{figure}

Figure~\ref{twoDomTree-R2} illustrates two dominating trees rooted at $(0,0)$ and $(4,3)$ for $R=2$.

\begin{figure}[h!]
\centering
\includegraphics[width=2.7in]{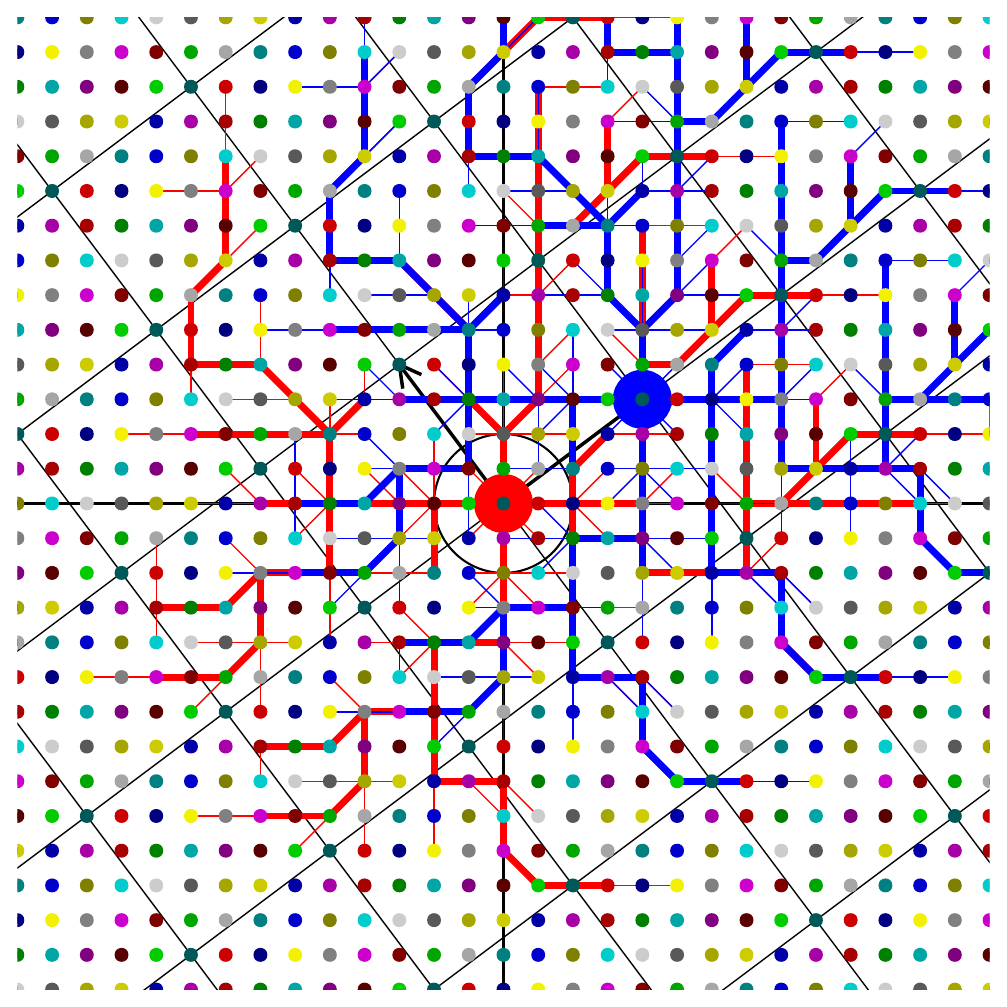}
\caption{Two dominating trees generated by ORCHID for R=2.\label{twoDomTree-R2}}
\end{figure}


Figure~\ref{allHighways-R2} illustrates the highways built by ORCHID for $R=2$ and $R=3$. 

\begin{figure}[h!]
\centering
\subfigure [R=2]{\includegraphics[width=2.7in]{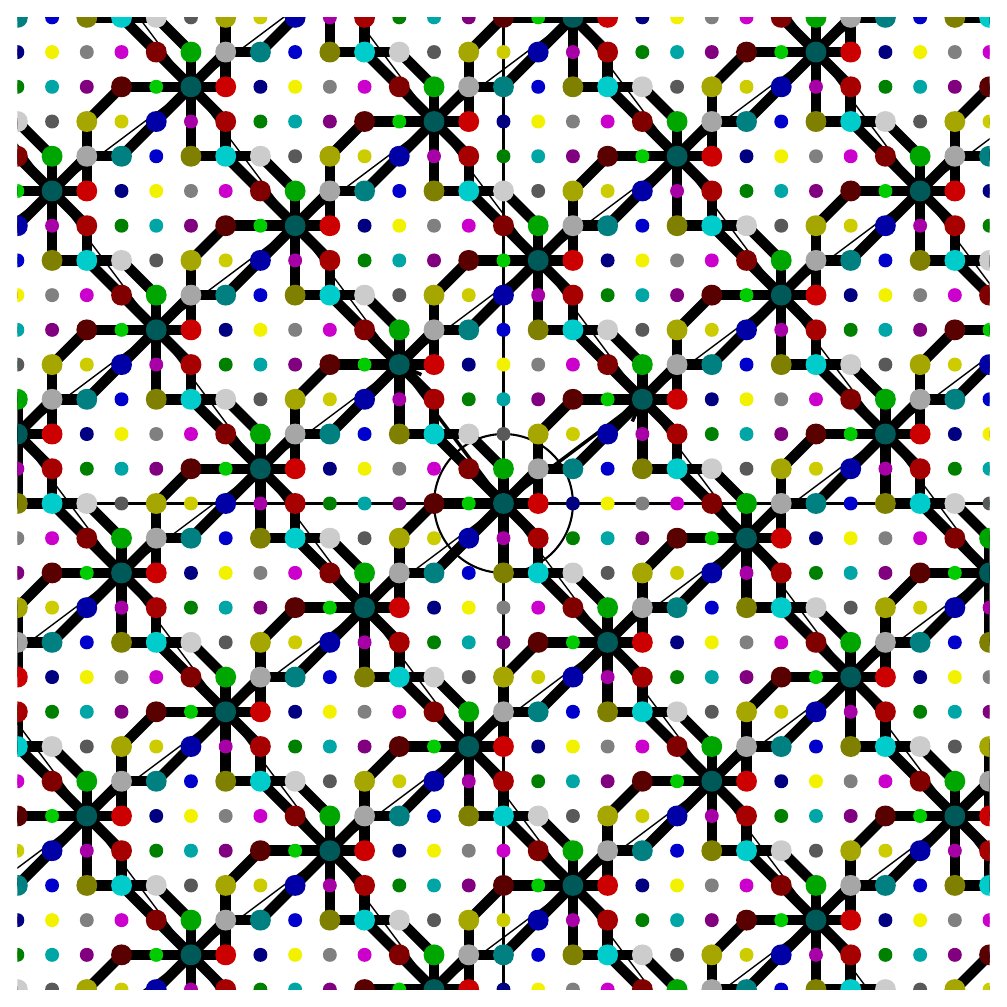}}
\subfigure [R=3]{\includegraphics[width=2.7in]{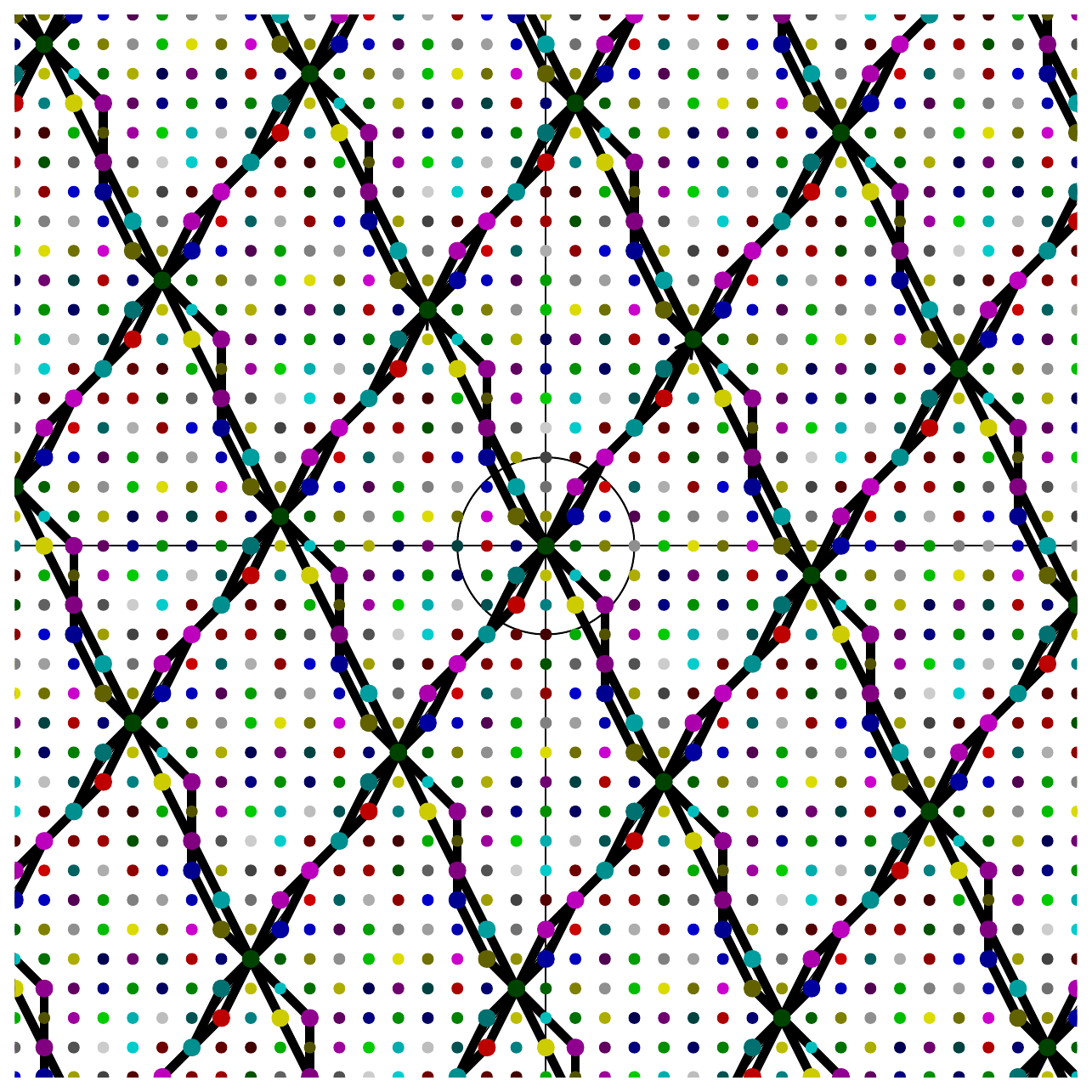}}
\caption{Highways built by ORCHID.
\label{allHighways-R2}}
\end{figure}

\newpage
Figure~\ref{short-path-delay-R2} illustrates the shortest-delay path between the node $(11,11)$ and the sink at $(0,0)$, and also the greedy routing path from $(-11,-11)$ to the same sink. The numbers in the circles indicate the color
of each node.

\begin{figure}[h!]
\centering
\includegraphics[width=3.5in]{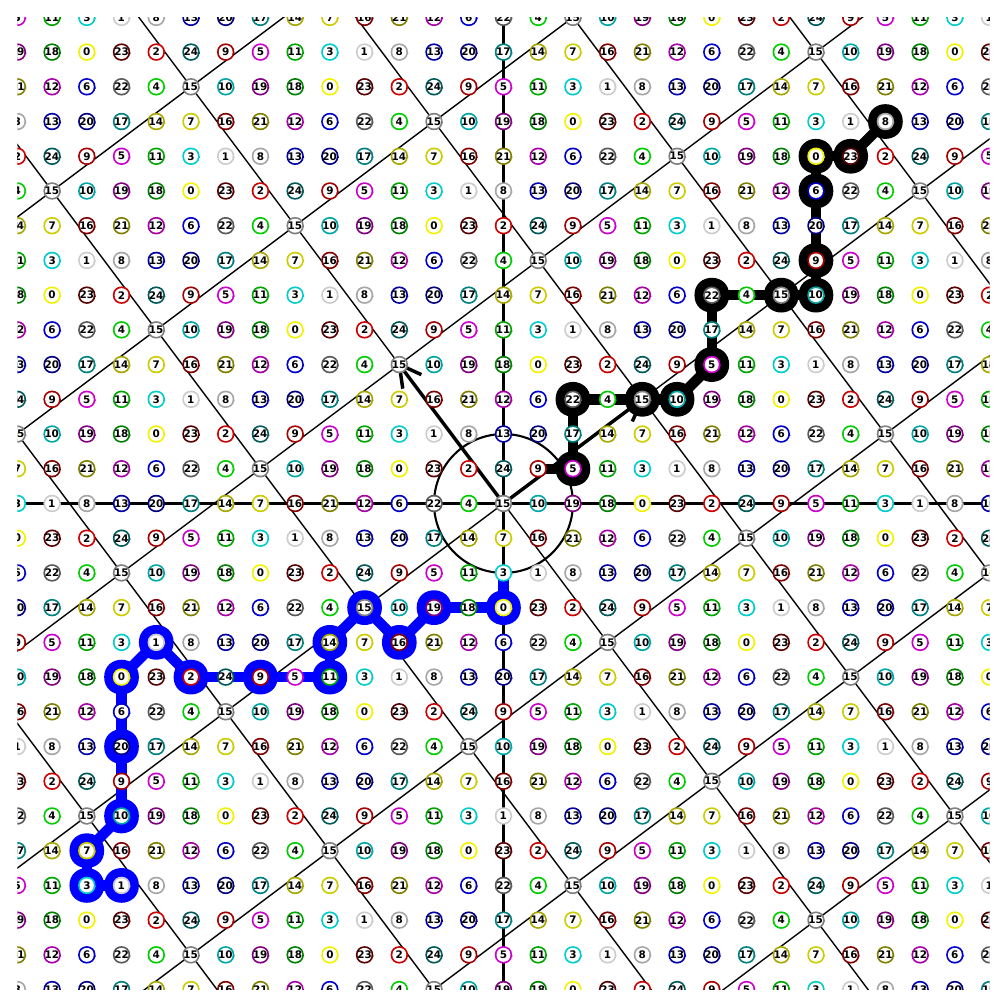}
\caption{Examples of paths computed for $R=2$.
\label{short-path-delay-R2}}
\end{figure}

\end{appendix}

\end{document}